\documentclass{iopjournal}

\makeatletter
\renewcommand{\articletype}[1]{%
  \vspace*{-8mm}%
  \noindent{\scriptsize\sffamily\bfseries\MakeUppercase{#1}}%
  \par\vspace{8mm}%
}

\renewcommand{\headrulewidth}{0pt}
\renewcommand{\footrulewidth}{0pt}
\makeatother

\usepackage{amssymb}
\usepackage{amsmath}
\usepackage{framed}
\usepackage{bm}
\usepackage{cite}
\usepackage{comment}
\usepackage{algorithmic}
\usepackage{algorithm}
\usepackage{url}
\usepackage{booktabs}
\usepackage{array}
\usepackage{braket}
\usepackage{titlesec}

\usepackage{mathtools}

\usepackage{amsfonts}
\usepackage{color}
\usepackage{cases}
\usepackage{enumerate}

\newtheorem{theorem}{Theorem}
\newtheorem{assumption}{Assumption}

\newtheorem{remark}{Remark}

\newtheorem{proposition}{Proposition}

\newtheorem{proof}{Proof}

\usepackage{blkarray}

\newcommand{\rd}{\mathrm{d}}

\newcommand{\RR}{\mathbb{R}}

\newcommand{\CC}{\mathbb{C}}

\usepackage{blkarray}

\usepackage{latexsym}

\usepackage{enumitem} 
\newlist{assump}{enumerate}{1}
\setlist[assump]{label=A\arabic*, ref=A\arabic*, leftmargin=*, labelindent=0pt}

\usepackage{cite}
\usepackage{url}

\begin{document}

\articletype{Paper} 
%
\title{Controllability-Based Input Allocation in Quantum Systems}

\author{
Riku Yaegashi$^{1,*}$,
Kazuhiro Sato$^1$
and Kentaro Ohki$^2$
}

\affil{$^1$Department of Mathematical Informatics,
Graduate School of Information Science and Technology,
The University of Tokyo,
Tokyo 113-8656, Japan}

\affil{$^2$Department of Applied Computer Engineering,
School of Information Science and Technology,
Tokai University,
Kanagawa 259-1292, Japan}

\affil{$^*$Author to whom any correspondence should be addressed.}

\email{yaegashi-riku184@g.ecc.u-tokyo.ac.jp}

\begin{abstract}

Understanding how local perturbations influence global quantum dynamics is important
for analyzing and designing large-scale quantum systems.
Existing quantum centrality measures often rely on static quantities such as correlations or entanglement and therefore
do not directly capture the dynamical propagation of perturbations under unitary evolution.
In this paper, we formulate controllability allocation scoring for quantum input-output dynamics governed by the Schrödinger equation.
Building on the multiobservable quantum-control Gramian, we introduce virtual local Hamiltonian perturbations along a nominal trajectory and resolve the terminal-output response into channel-wise Gramian contributions.
These contributions are combined through a simplex-constrained log-determinant optimization, yielding a volumetric controllability allocation score (V-CAS) that specifies an optimal allocation of virtual input intensities among local Hamiltonian perturbation directions. 
The resulting allocation provides an output- and trajectory-dependent measure of their cooperative dynamical influence.
We characterize the range and minimum-energy properties of the resulting output Gramian and derive a drift-independent sufficient condition that guarantees uniqueness of the score for almost every terminal time at which the scoring problem is feasible.
For quantum many-body systems, we develop an implementation based on matrix product states and time-evolving block decimation that avoids explicitly constructing exponentially large state-transition matrices.
We further derive an error bound that propagates Trotter and MPS-truncation errors to the discrete output Gramian.
Numerical experiments on quantum many-body models show that the resulting score captures finite-time influence patterns that can differ from those identified by static correlation- and entanglement-based measures and can provide useful information for target-independent control-channel preselection.

\end{abstract}

\section{Introduction}\label{sec:introduction}
As quantum systems scale in size and complexity, understanding how global dynamics emerge from local interactions becomes increasingly important.
This issue is particularly relevant in many-body settings, where available control and measurement channels are often limited.
It is therefore useful to identify which local subsystems or Hamiltonian directions exert strong dynamical influence during the system's natural evolution.
However, the correlation- and entanglement-based quantum centrality measures considered in this work \cite{sundar2018complex,nokkala2024complex} are evaluated from static state properties and do not directly quantify how local Hamiltonian perturbations propagate under unitary dynamics.

To address this limitation, we consider a controllability-based allocation of virtual control inputs.
We introduce local Hamiltonian perturbations as virtual control directions and determine their optimal allocation based on the terminal-output variations that they generate along a nominal trajectory.
The resulting allocation depends on the initial state, the terminal time, the perturbation Hamiltonians, and the selected observables, and therefore characterizes trajectory-dependent dynamical influence rather than a static property of the drift Hamiltonian alone. 

Our construction is motivated by output controllability analysis in classical control theory.
Output controllability Gramians characterize reachable terminal-output directions and the input energy required to generate them \cite{kailath1980linear}.
Controllability-scoring formulations based on weighted Gramians and log-determinant optimization have also been developed to quantify the cooperative contribution of multiple input channels to the reachable set \cite{sato2024controllability}.
More recently, the controllability allocation score (CAS) was introduced to determine the optimal allocation of intervention intensities among prescribed input directions with respect to selected outputs \cite{sato2025target}.
In quantum control, a multiobservable quantum-control Gramian has been introduced to characterize the first-order responses of multiple terminal observables to control variations \cite{chakrabarti2008quantum}.

Extending this controllability allocation scoring framework to quantum systems requires several issues specific to the quantum input-output setting to be addressed.
First, physical pure states are defined only up to a global phase, whereas the linearized Schr\"odinger dynamics is naturally represented using state vectors.
The resulting sensitivity analysis must therefore be connected consistently to phase-invariant observable outputs.
Second, the optimal allocation of virtual input intensities need not be unique.
Since the allocation weights themselves define the CAS, such nonuniqueness may lead to ambiguous rankings of local Hamiltonian directions and potentially different control-channel selections based on the scores. 
This motivates the derivation of conditions that guarantee uniqueness of the optimal allocation.
Finally, in many-body systems, direct evaluation of the required sensitivities becomes computationally prohibitive because of the exponential growth of the Hilbert-space dimension.
In this paper, we formulate controllability allocation scoring for quantum input-output dynamics by combining the multiobservable quantum-control Gramian with local virtual Hamiltonian perturbations along a nominal Schr\"odinger trajectory.
The terminal-output Gramian is resolved into contributions associated with the individual perturbation directions, and these contributions are combined through a simplex-constrained log-determinant optimization.
The resulting optimal allocation specifies how virtual input intensities should be distributed among local Hamiltonian perturbation directions to maximize the volume of the first-order reachable set in the prescribed observable space.

We characterize the range and minimum-energy properties of the resulting output Gramian.
We then derive a sufficient condition for uniqueness of the controllability allocation score that depends only on the initial state, the local perturbation Hamiltonians, and the differential of the output map at the initial state.
In particular, the condition is independent of the drift Hamiltonian and guarantees affine independence of the finite-time Gramian contributions, and hence strict convexity of the scoring problem, for almost every terminal time.

For quantum many-body systems, we employ matrix product states to avoid explicitly constructing exponentially large state-transition matrices; this representation generalizes the product-state ansatz used in mean-field theory.The required sensitivities are evaluated through forward propagation of the nominal state, backward propagation of observable states, local operator applications, and MPS inner-product contractions, without explicitly constructing the full state-transition matrix.
We further quantify how Trotter and MPS-truncation errors propagate to the resulting discrete output Gramian.

The main contributions are as follows:
\begin{itemize}
    \item We formulate controllability allocation scoring for quantum input-output dynamics by resolving the first-order terminal-observable response along a nominal Schrödinger trajectory into contributions associated with local Hamiltonian perturbation directions. The resulting optimal allocation is output-dependent and trajectory-dependent and is defined through physically phase-invariant observables.    
    \item We derive a drift-independent sufficient condition, expressed only in terms of the initial state, the local perturbation Hamiltonians, and the initial output differential, that guarantees restricted injectivity for almost every terminal time and, whenever the allocation problem is feasible, uniqueness of the controllability allocation score.
    \item We develop an MPS/TEBD-compatible implementation for evaluating the channel-resolved output Gramian without explicitly constructing exponentially large state-transition matrices, together with an error bound that propagates Trotter and MPS-truncation errors to the discrete Gramian.
\end{itemize}

The paper is structured as follows.
We begin in Section~\ref{sec:Prelim} by formulating the controlled quantum input-output dynamics and introducing the output controllability Gramian together with its basic properties.
Section~\ref{sec:quantum-controllability} formulates the controllability allocation scoring problem for local Hamiltonian perturbations, establishes its mathematical properties, including a sufficient condition for uniqueness, and describes the numerical algorithm used to solve the allocation problem.
Section~\ref{sec:MPS} presents the MPS/TEBD-based implementation for quantum many-body systems and derives an error bound for the resulting Gramian approximation.
In Section~\ref{sec:numerical}, we numerically compare the proposed score with static quantum centralities, examine its use for control-channel preselection, and evaluate the accuracy and computational cost of the MPS implementation.
Finally, Section~\ref{sec:conclusion} summarizes the main results and discusses future directions.
Technical proofs, error estimates, and detailed numerical settings are provided in the appendices.

\textbf{Notation.}
We denote by $\mathbb R$ and $\mathbb C$ the fields of real and complex numbers, respectively. For an inner product space $X$, we denote its inner product by $\langle\cdot,\cdot\rangle_X$.

For a real vector or matrix $A$, we denote by $A^\top$ its transpose. 
For a complex vector or matrix $A$, we denote by $A^\dagger$ its Hermitian conjugate.

Let \(\mathcal H\) be a finite-dimensional complex Hilbert space equipped with the Hermitian inner product \(\langle\cdot,\cdot\rangle_{\mathcal H}\). The associated norm is defined by
\[
    \|\psi\|_{\mathcal H}
    :=
    \sqrt{\langle\psi,\psi\rangle_{\mathcal H}}.
\]

For Hilbert spaces \(\mathcal H_1\) and \(\mathcal H_2\), we denote by \(\mathcal L(\mathcal H_1,\mathcal H_2)\) the space of linear maps from \(\mathcal H_1\) to \(\mathcal H_2\). In particular, we write
\[
    \mathcal L(\mathcal H):=\mathcal L(\mathcal H,\mathcal H).
\]
An element of \(\mathcal L(\mathcal H)\) is called an operator on \(\mathcal H\).

When the system is composed of \(n\) local subsystems, we write
\[
    \mathcal H
    =
    \bigotimes_{\ell=1}^n \mathcal H_\ell ,
\]
where \(\mathcal H_\ell\) is the local Hilbert space at site \(\ell\).
Throughout this paper, \(\otimes\) denotes the Hilbert-space tensor product.

For a local operator \(A\in\mathcal L(\mathcal H_\ell)\), we denote by
\(A^{(\ell)}\in\mathcal L(\mathcal H)\) its extension to the total Hilbert
space, defined by
\[
    A^{(\ell)}
    :=
    I_1\otimes\cdots\otimes I_{\ell-1}
    \otimes A
    \otimes I_{\ell+1}\otimes\cdots\otimes I_n ,
\]
where \(I_r\in\mathcal L(\mathcal H_r)\) is the identity operator on
\(\mathcal H_r\). We call \(X\in\mathcal L(\mathcal H)\) a single-site local
operator if there exist a site \(\ell\in\{1,\ldots,n\}\) and an operator
\(A\in\mathcal L(\mathcal H_\ell)\) such that $    X=A^{(\ell)}.$
Unless otherwise stated, all states and operators are understood as objects on
the total Hilbert space \(\mathcal H\). The tensor-product structure is used
to specify the locality of operators.

Since we consider pure states in a finite-dimensional Hilbert space \(\mathcal H\), we fix an orthonormal basis and identify \(\mathcal H\) with \(\mathbb C^d\), where \(d=\dim\mathcal H\). Under this identification, the Hermitian inner product is represented by
\[
\langle z,w\rangle_{\mathbb C^d}:=z^\dagger w,
\qquad z,w\in\mathbb C^d.
\]
Physical pure states can then be identified with elements of the complex projective space \(\mathbb{CP}^{d-1}\).

To distinguish normalized state vectors from physical states, we define
\[
S_{\mathbb C^d}:=\{z\in\mathbb C^d\mid\|z\|=1\},
\qquad \|z\|:=\sqrt{z^\dagger z},
\]
and
\begin{gather}
    \label{eq:quantum-state-quotient}
    \mathcal M:=S_{\mathbb C^d}/U(1)\simeq\mathbb{CP}^{d-1},    
\end{gather}
where \(U(1):=\{e^{i\theta}\mid\theta\in\mathbb R\}\). The canonical quotient map is defined by
\[
\pi:S_{\mathbb C^d}\to\mathcal M,
\qquad \pi(\psi)=[\psi],
\]
where \([\psi]:=\{e^{i\theta}\psi\mid\theta\in\mathbb R\}\).
Throughout this paper, \(\psi\in S_{\mathbb C^d}\) denotes a normalized state vector, whereas \(x=[\psi]\in\mathcal M\) denotes the corresponding physical pure state.

With respect to the preceding orthonormal basis, every operator \(A\in\mathcal L(\mathcal H)\) is represented by a matrix in \(\CC^{d\times d}\). In particular, a self-adjoint operator on \(\mathcal H\) is represented by a Hermitian matrix.

We write \(\mathrm{i}\) for the imaginary unit.
For a complex number \(z=x+\mathrm{i}y\in\CC\), \(\operatorname{Re}(z)=x\) denotes its real part.

We regard \(\mathcal S_{\CC^d}\) as a real smooth manifold. For \(z\in\mathcal S_{\CC^d}\), the tangent space \(\mathcal T_z\mathcal S_{\CC^d}\) is the set of velocities at \(t=0\) of smooth curves on \(\mathcal S_{\CC^d}\) passing through \(z\). 

We denote by \(L^2([0,T];\RR^m)\) the space of square-integrable
\(\RR^m\)-valued measurable functions on \([0,T]\), equipped with the
standard inner product
\[
    \langle u,v\rangle_{L^2}
    :=
    \int_0^T u(t)^\top v(t)\,\rd t .
\]

For inner product spaces \(X\) and \(Y\) and a linear operator
\(A:X\to Y\), the adjoint operator \(A^*:Y\to X\) is defined by
\[
    \langle Ax,y\rangle_Y
    =
    \langle x,A^*y\rangle_X,
    \qquad
    x\in X,\ y\in Y.
\]

We denote by $\mathrm{Sym}_r$ the space of real symmetric $r\times r$ matrices. For a symmetric matrix \(A=(a_{ij})\in\mathrm{Sym}_r\), we denote by
\[
    \operatorname{vech}(A)\in\mathbb R^{r(r+1)/2}
\]
the half-vectorization of \(A\), obtained by stacking the entries on and
below the main diagonal columnwise:
\[
    \operatorname{vech}(A)
    :=
    \begin{bmatrix}
        a_{11} &
        a_{21} &
        \cdots &
        a_{r1} &
        a_{22} &
        a_{32} &
        \cdots &
        a_{rr}
    \end{bmatrix}^{\top}.
\]

\section{Preliminaries} \label{sec:Prelim}

In this section, we formulate the controlled quantum dynamics and introduce the linearized dynamics along a nominal trajectory. 
Based on this linearized system, we define the controllability Gramian that will be used to construct the controllability allocation score.

\subsection{Quantum System Dynamics}
In this subsection, we formulate a controlled quantum input-output system and derive its linearization along a reference trajectory.

We first introduce the Schrödinger dynamics using normalized
state-vector representatives. Let
\(\psi(t)\in\mathcal S_{\CC^d}\) satisfy
\begin{equation}
    \mathrm{i}\frac{\rd}{\rd t}\psi(t)
    =
    H(t)\psi(t),
    \qquad
    \psi(0)=\psi_0\in\mathcal S_{\CC^d},
    \label{eq:schrodinger}
\end{equation}
where we set \(\hbar=1\).

We next specify the dependence of the Hamiltonian on the control
inputs. The Hamiltonian is decomposed into a drift term \(H_0\) and
control terms as
\begin{equation}
    H(t)
    =
    H_0+\sum_{j=1}^{m}u_j(t)H_j,
    \label{eq:hamiltonian}
\end{equation}
where \(H_0,H_j\in\CC^{d\times d}\) are Hermitian matrices and \(u_j(t)\in\RR\) are real-valued control inputs.
Throughout this paper, we assume that $u\in L^2([0,T];\mathbb R^m)$.
Under this assumption, the Schrödinger equation admits a unique absolutely continuous solution and a unitary propagator $\Phi_u(t,s)$, ensuring norm preservation. 
The proof is provided in Appendix \ref{proof:well-posed-AC}

We then introduce physical outputs given by expectation values of prescribed observables. In many quantum applications, the quantities of interest are not the full quantum state but a finite collection of measurable quantities. Expectation values of selected observables provide a standard class of such outputs
\cite{zhang2017observation}.

Let \(O_1,\ldots,O_r\in\CC^{d\times d}\) be Hermitian observables. We
define the physical output map $    \mathcal O:\mathcal M\to\RR^r$ by
\[
    \mathcal O([\psi])
    :=
    \begin{bmatrix}
        \psi^\dagger O_1\psi\\
        \vdots\\
        \psi^\dagger O_r\psi
    \end{bmatrix}.
\]
This map is well-defined on \(\mathcal M\), since
\[
    (e^{\mathrm{i}\theta}\psi)^\dagger
    O_k
    (e^{\mathrm{i}\theta}\psi)
    =
    \psi^\dagger O_k\psi,
    \qquad
    \theta\in\RR,
    \quad
    k=1,\ldots,r.
\]

For the variational analysis, we use the corresponding representation
on normalized state vectors,
\[
    \widetilde{\mathcal O}:
    \mathcal S_{\CC^d}\to\RR^r,
    \qquad
    \widetilde{\mathcal O}(\psi)
    :=
    \mathcal O([\psi]).
\]
By construction,
\begin{equation}
    \widetilde{\mathcal O}(e^{\mathrm{i}\theta}\psi)
    =
    \widetilde{\mathcal O}(\psi),
    \qquad
    \theta\in\RR.
    \label{eq:output-invariande}
\end{equation}

Combining the Schrödinger dynamics with the physical state and output
maps, the nonlinear quantum input-output system considered in this
paper is
\begin{equation}
    \begin{aligned}
        \frac{\rd}{\rd t}\psi(t)
        &=
        -\mathrm{i}
        \left(
            H_0+\sum_{j=1}^{m}u_j(t)H_j
        \right)\psi(t),
        \\
        x(t)
        &=
        [\psi(t)]\in\mathcal M,
        \\
        y(t)
        &=
        \mathcal O(x(t))
        =
        \widetilde{\mathcal O}(\psi(t)).
    \end{aligned}
    \label{eq:main_dynamics}
\end{equation}
Here, \(\psi(t)\in\mathcal S_{\CC^d}\) is a normalized representative
used to describe the Schrödinger dynamics, whereas
\(x(t)=[\psi(t)]\) is the corresponding physical pure state.

We next introduce the linearized system. We fix a reference control and define the corresponding nominal trajectory. Let \(\bar u(\cdot)\) be a prescribed control with components \(\bar u_1(t),\ldots,\bar u_m(t)\). The corresponding normalized state-vector trajectory \(\bar\psi(t)\) satisfies
\begin{equation}
    \frac{\rd}{\rd t}\bar\psi(t)
    =
    -\mathrm{i}
    \left(
        H_0+\sum_{j=1}^{m}\bar u_j(t)H_j
    \right)
    \bar\psi(t).
    \label{eq:nominal_dynamics}
\end{equation}
We refer to \(\bar\psi(\cdot)\) as the nominal state-vector trajectory.
The associated physical state and output are
\[
    \bar x(t)
    :=
    [\bar\psi(t)]\in\mathcal M,
    \qquad
    \bar y(t)
    :=
    \mathcal O(\bar x(t))
    =
    \tilde{\mathcal O}(\bar\psi(t)).
\]

We finally derive the first-order input-output response around the
reference control. Let
\[
    \delta u(t)
    =
    \bigl(
        \delta u_1(t),\ldots,\delta u_m(t)
    \bigr)
\]
be a small control variation, and let
\(\delta\psi(t)\) and \(\delta y(t)\) denote the corresponding
first-order state-vector and output variations. Since the initial
state is fixed, we impose
\[
    \delta\psi(0)=0.
\]
The linearized input-output system is
\begin{equation}
    \begin{aligned}
        \dot{\delta\psi}(t)
        &=
        -\mathrm{i}
        \left(
            H_0+\sum_{k=1}^{m}\bar u_k(t)H_k
        \right)\delta\psi(t)
        +
        \sum_{j=1}^{m}
        \delta u_j(t)(-\mathrm{i}H_j)\bar\psi(t),
        \\
        \delta y(t)
        &=
        C(t)[\delta\psi(t)],
    \end{aligned}
    \label{eq:linearized_dynamics}
\end{equation}
where
\[
    C(t)
    :=
    D\widetilde{\mathcal O}_{\bar\psi(t)}
    :
    T_{\bar\psi(t)}\mathcal S_{\CC^d}
    \to\RR^r
\]
is the differential of the lifted output map along the nominal  trajectory.
For the expectation-value output considered in this paper,
the differential $C(t)$ admits the explicit representation
\begin{align}
    [C(t)[v]]_\ell
    &=
    D\widetilde{\mathcal O}_{\ell,\bar\psi(t)}[v]=
    v^\dagger O_\ell \bar\psi(t)
    +
    \bar\psi(t)^\dagger O_\ell v=
    2\operatorname{Re}
    \left\langle
        O_\ell\bar\psi(t),\,v
    \right\rangle,
    \qquad
    \ell=1,\ldots,r,
    \label{eq:explicit_output_differential}
\end{align}
for $v\in T_{\bar\psi(t)}\mathcal S_{\CC^d}$, where the last equality
follows from the Hermiticity of $O_\ell$.
Equivalently,
\begin{align}
    C(t)[v]
    =
    2
    \begin{bmatrix}
        \operatorname{Re}
        \langle O_1\bar\psi(t),v\rangle
        \\
        \vdots
        \\
        \operatorname{Re}
        \langle O_r\bar\psi(t),v\rangle
    \end{bmatrix}.
\end{align}

Since the quantum state is defined as a quotient space in \eqref{eq:quantum-state-quotient}, a physical output must be independent of the choice of representative. Accordingly, the lifted output map \(\widetilde{\mathcal O}\) satisfies \eqref{eq:output-invariande}
and hence
\[
    D\widetilde{\mathcal O}_{\psi}[i\psi]=0.
\]
Therefore, the global-phase direction does not contribute to the terminal-output Gramian.

\subsection{Output Controllability Gramian}
\label{subsec:output-gramian}

In this subsection, we define the output controllability Gramian from the terminal output variations generated by the linearized input-output system.

We first introduce the state-transition operator associated with the nominal dynamics. Let \(\Phi_{\bar u}(t,s)\) be defined by
\[
    \frac{\partial}{\partial t}\Phi_{\bar u}(t,s)
    =
    -\mathrm{i}
    \left(
        H_0+\sum_{k=1}^{m}\bar u_k(t)H_k
    \right)
    \Phi_{\bar u}(t,s),
    \qquad
    \Phi_{\bar u}(s,s)=I_d.
\]
Then the nominal trajectory $\bar{\psi}$ satisfying \eqref{eq:nominal_dynamics} satisfies
\[
    \bar\psi(t)
    =
    \Phi_{\bar u}(t,s)\bar\psi(s).
\]

We next compute the terminal state-vector variation. Since the initial
state is fixed, we impose \(\delta\psi(0)=0\). By the
variation-of-constants formula, the solution of
\eqref{eq:linearized_dynamics} at time \(T\) is
\begin{equation}
    \delta\psi(T)
    =
    \sum_{j=1}^{m}
    \int_0^T
    \Phi_{\bar u}(T,t)
    (-\mathrm{i}H_j)
    \bar\psi(t)
    \delta u_j(t)\,\rd t.
    \label{eq:linearized-solution}
\end{equation}

We then obtain the corresponding terminal output variation. 
Applying the output differential
\[
    C(T)
    =
    D\tilde{\mathcal O}_{\bar\psi(T)}
\]
to \eqref{eq:linearized-solution} gives
\begin{equation}
    \delta y(T)
    =
    \sum_{j=1}^{m}
    \int_0^T
    C(T)
    \left[
        \Phi_{\bar u}(T,t)
        (-\mathrm{i}H_j)
        \bar\psi(t)
    \right]
    \delta u_j(t)\,\rd t.
    \label{eq:terminal-output-variation}
\end{equation}

This formula defines a linear operator from input variations to terminal output variations:
\begin{equation}
    J_y:
    L^2([0,T];\RR^m)
    \longrightarrow
    \RR^r,
    \qquad
    J_y[\delta u]
    :=
    \delta y(T).
    \label{eq:output-endpoint-operator}
\end{equation}
In other words, \(J_y\) maps an infinitesimal input variation
\(\delta u(\cdot)\) to the corresponding first-order output variation
at time \(T\).

We define the reachable subspace of terminal output variations by
\[
    \mathcal R_y(T)
    :=
    \left\{
        J_y[\delta u]
        \ \middle|\
        \delta u\in L^2([0,T];\RR^m)
    \right\}
    \subseteq\RR^r.
\]
The linearized input-output system is said to be output controllable
at time \(T\) along the nominal trajectory if
\[
    \mathcal R_y(T)=\RR^r.
\]

We now define the output controllability Gramian by
\begin{equation}
    W(T)
    :=
    J_yJ_y^*
    :
    \RR^r
    \longrightarrow
    \RR^r,
    \label{eq:output-gramian-operator}
\end{equation}
where \(J_y^*\) denotes the adjoint of \(J_y\) with respect to the
standard Euclidean inner product on \(\RR^r\) and the standard
\(L^2\) inner product on \(L^2([0,T];\RR^m)\).

The output controllability Gramian admits the explicit representation
\begin{equation}
\begin{aligned}
    W(T)
    =
    \sum_{j=1}^{m}
    \int_0^T
    &C(T)
    \left[
        \Phi_{\bar u}(T,t)
        (-\mathrm{i}H_j)
        \bar\psi(t)
    \right]
    \left( C(T) 
    \left[
        \Phi_{\bar u}(T,t)
        (-\mathrm{i}H_j)
        \bar\psi(t)
    \right]\right)^\top
    \,\rd t.
\end{aligned}
\label{eq:output-gramian}
\end{equation}
We call $W(T)$ the output controllability Gramian along the nominal trajectory $\bar\psi(\cdot)$. 
When $W(T)$ is positive definite, the minimum input energy required to realize a terminal-output variation $z\in\mathbb R^r$ is given by
\[
z^\top W(T)^{-1}z=\|\widehat{\delta u}\|_{L^2}^2,
\]
where $\widehat{\delta u}$ denotes the minimum-energy input. 
Thus, the inverse Gramian characterizes how difficult it is to realize variations in different output directions.

The following proposition characterizes output controllability in terms of the output Gramian.

\begin{proposition}
\label{prop:image-output-gramian}
The image of the output controllability Gramian coincides with the reachable subspace of terminal output variations:
\begin{equation}
    \operatorname{Im}W(T)
    =
    \mathcal R_y(T).
\end{equation}
Consequently, the following statements are equivalent:
\begin{enumerate}
    \item $W(T)$ is positive definite.
    \item $J_y$ is surjective onto $\mathbb R^r$.
    \item The linearized input-output system is output controllable at
    time $T$ along the nominal trajectory.
\end{enumerate}
\end{proposition}
\begin{proof}
    See Appendix \ref{proof:image-output-gramian}.
\end{proof}
Proposition \ref{prop:image-output-gramian} shows that the reachable subspace of terminal-output variations can be characterized by the output Gramian. 
In particular, whether arbitrary first-order terminal-output variations are achievable can be determined simply by checking the invertibility of $W(T)$.
\subsection{Minimum-Energy Property of the Output Gramian}
\label{subsec:minimum-energy-output}

In this subsection, we show that the output controllability Gramian
characterizes the minimum input energy required to generate a
prescribed terminal output variation.

For a reachable terminal output variation
\[
    z\in\mathcal R_y(T)=\operatorname{Im}J_y,
\]
define the minimum input energy by
\begin{equation}
    E_{\min}^{y}(z)
    :=
    \inf
    \left\{
        \|\delta u\|_{L^2}^2
        \ \middle|\
        J_y[\delta u]=z
    \right\}.
    \label{eq:minimum-output-energy}
\end{equation}

The following proposition gives the minimum-energy input in terms of the
output controllability Gramian.

\begin{proposition}[Minimum-energy property]\label{thm:minimum-output-energy}
Let $z\in\mathcal R_y(T)$. Then the unique minimum-norm input variation satisfying
\begin{equation}
    J_y[\delta u]=z
\end{equation}
is
\begin{equation}
    \delta u^\star
    =
    J_y^*W(T)^+z,
\end{equation}
where $W(T)^+$ denotes the Moore--Penrose pseudoinverse of $W(T)$.
Moreover, the minimum energy is
\begin{equation}
    E_{\min}^y(z)
    =
    z^\top W(T)^+z.
    \label{eq:minimum-energy-property}
\end{equation}
\end{proposition}
\begin{proof}
    See Appendix \ref{app:min-output-energy-proof}.
\end{proof}

\section{Controllability Allocation Scoring}\label{sec:quantum-controllability}
\subsection{Formulation of the Controllability Allocation Scoring Problem}

We now formulate controllability allocation scoring for local Hamiltonian perturbations along the natural evolution of the quantum system. 
The construction follows the controllability allocation scoring framework of \cite{sato2025target}, while the Gramian contributions are defined here from the terminal-output responses of the quantum system introduced in the preceding section.

In this section, we consider the uncontrolled Schr\"odinger dynamics
\begin{equation}\label{target}
    \frac{\rd}{\rd t}\psi(t)
    =
    -\mathrm{i}H_0\psi(t),
\end{equation}
which defines the natural evolution of the system. 
To formulate the allocation of virtual input intensities among local perturbation directions, we introduce virtual control Hamiltonians 
$\{H_j\}_{j=1}^{m}$ and consider the augmented system
\begin{equation}\label{eq:virtual_input}
    \frac{\rd}{\rd t}\psi(t)
    =
    -\mathrm{i}
    \left(
        H_0
        +
        \sum_{j=1}^{m}\sqrt{p_j}\,u_j(t)H_j
    \right)
    \psi(t),
\end{equation}
where $p_j \geq 0$ represents the relative virtual-input intensity allocation assigned to the $j$-th virtual perturbation direction.
The inputs $u_j$ are introduced only to probe the local response of the system, while the nominal trajectory used to construct the score remains the uncontrolled evolution \eqref{target}.

The objective is to determine the weight vector $p$ so as to maximize the first-order terminal-output controllability of the augmented system along the uncontrolled trajectory.
In this setting, larger controllability corresponds to a larger first-order reachable set in the prescribed output space under a common input-energy budget.
After optimization, the resulting weights specify an optimal allocation of virtual-input intensities among the perturbation directions for enlarging this reachable set.
Thus, a larger optimal weight $\hat{p}_j$ indicates that the corresponding Hamiltonian direction $H_j$ share of the jointly optimized allocation, rather than measuring its influence in isolation.

To determine the optimal weights, we next construct the corresponding weighted output controllability Gramian and introduce the associated log-determinant criterion.
\begin{remark}[Physical scaling of the virtual control directions]
    In the augmented system~\eqref{eq:virtual_input}, each \(H_j\) is interpreted as the Hamiltonian generated by a unit input through channel \(j\), and the input space is equipped with the common \(L^2\) energy metric. 
    Therefore, the magnitude of \(H_j\) is reflected in the optimal allocation. 
    If one channel generates a stronger Hamiltonian under the same input-energy budget, its Gramian contribution is correspondingly larger, and the channel may receive a larger allocation weight.
    This dependence is intentional: the proposed score measures effective dynamical influence per unit input energy, rather than an intrinsic structural property of a site independent of the available actuation.
    A rescaling of \(H_j\) that represents only a change of parametrization must be accompanied by the inverse rescaling of the corresponding input variable and its energy metric.
\end{remark}
\begin{remark}
    Although the definition allows general virtual Hamiltonian directions $H_j$, the interpretation of the resulting allocation score as a local dynamical centrality is clearest when the virtual Hamiltonians are local.
    Moreover, local Hamiltonian models and locally accessible operations are standard settings in quantum simulation and quantum control \cite{georgescu2014quantum,koch2022quantum}. In addition, controllability of quantum networks through the manipulation of a local Hamiltonian on a small subsystem has been studied as a standard local-control setting \cite{burgarth2009local}. 
    Therefore, in the following, we assume that the virtual Hamiltonian directions $H_j$ are local operators to clarify the interpretation of the controllability allocation score as a local dynamical centrality.
\end{remark}

Throughout this paper, we assume the following.
\begin{assumption}
\label{assump:local-control-hamiltonian}
 Each virtual input Hamiltonian \(H_j\) is a local operator acting nontrivially only on a single site. 
\end{assumption}

\subsection{Metric of Controllability}
To evaluate the controllability of the system \eqref{eq:virtual_input}, we consider the output controllability Gramian along the nominal trajectory $\bar u = 0$. 
Here, $\bar u = 0$ corresponds to the unperturbed evolution of the original system \eqref{target}, and the virtual inputs $u_j$ are introduced only to probe the local response of the dynamics to perturbations along the directions defined by $H_j$.

Based on this geometric interpretation, we adopt the following metric.
In the preceding discussion, we wrote the Gramian simply as $W(T)$.  In the following sections, however, we study its dependence on both the weight vector $p$ and the terminal time $T$.  Therefore, from now on, we write the output controllability Gramian as $W(p,T)$.  Here $W(p,T)$ is explicitly written as
\begin{equation}
    \begin{split}
        W(p,T) &=\sum_{j=1}^m\int_0^T C(T)\left[ \Phi (T,t) (-\mathrm{i}\sqrt{p}_jH_j)\bar \psi(t) \right] \left(C(T)\left[ \Phi (T,t) (-\mathrm{i}\sqrt{p}_jH_j)\bar \psi(t) \right]\right)^\top \rd t\\
        &=\sum_{j=1}^m p_j W_j(T)
    \end{split}
\end{equation}
where $\Phi$ is the state-transition operator associated with nominal dynamics
\begin{gather}
    \frac{\partial}{\partial t}\Phi(t,s) = -\mathrm{i} H_0 \Phi(t,s), \ \Phi(s,s)=I.
\end{gather}
and \(W_j(T)\) denotes the contribution of the \(j\)-th virtual Hamiltonian direction to the output controllability Gramian.
\begin{equation}
        \begin{split}
     W_j(T)= \int_0^T
        C(T)\left[\Phi(T,t)(-\mathrm{i}H_j)\bar{\psi}(t)\right]\left(C(T)\left[\Phi(T,t)(-\mathrm{i}H_j)\bar{\psi}(t)\right]\right)^\top \rd t   
        \end{split}
\end{equation}
For notational convenience, define the output sensitivity associated with the
$j$-th virtual Hamiltonian direction by
\begin{equation}
    y_j(T,t)
    :=
    C(T)\left[
        \Phi(T,t)(-\mathrm{i}H_j)\bar{\psi}(t)
    \right]
    \in\mathbb{R}^r .
    \label{eq:output_sensitivity}
\end{equation}

Using the explicit form of the output differential for the
expectation-value outputs derived in Section~\ref{sec:Prelim}, the $\ell$-th component of
$y_j(T,t)$ is
\begin{equation}
\begin{aligned}
    [y_j(T,t)]_\ell
    &=
    2\operatorname{Re}
    \left\langle
        O_\ell\bar{\psi}(T),
        \Phi(T,t)(-\mathrm{i}H_j)\bar{\psi}(t)
    \right\rangle
    \\
    &=
    2\operatorname{Re}
    \left\langle
        \Phi(T,t)^*O_\ell\bar{\psi}(T),
        (-\mathrm{i}H_j)\bar{\psi}(t)
    \right\rangle,
    \qquad
    \ell=1,\ldots,r .
\end{aligned}
\label{eq:output_sensitivity_adjoint}
\end{equation}

This suggests introducing the adjoint state associated with the
$\ell$-th output observable,
\begin{equation}
    z_\ell(T,t)
    :=
    \Phi(T,t)^*O_\ell\bar{\psi}(T),
    \qquad
    \ell=1,\ldots,r .
    \label{eq:adjoint_state}
\end{equation}
Then, the channel-wise output sensitivity admits the representation
\begin{equation}
    [y_j(T,t)]_\ell
    =
    2\operatorname{Re}
    \left\langle
        z_\ell(T,t),
        (-\mathrm{i}H_j)\bar{\psi}(t)
    \right\rangle .
    \label{eq:adjoint_output_sensitivity}
\end{equation}

Accordingly, the channel-wise Gramian contributions can be evaluated from the nominal trajectory, the adjoint states $z_\ell(T,t)$, and their inner products with the virtual Hamiltonian directions, without explicitly constructing the state-transition operator. Section~\ref{sec:MPS} exploits this representation for an MPS/TEBD implementation in quantum many-body systems.

\subsection{Mathematical analysis}

Following the controllability allocation scoring framework of \cite{sato2025target}, we formulate the corresponding allocation problem for the quantum output Gramian introduced above.
The optimization problem is given by

\begin{equation}
\begin{aligned}
\min_{p \in \mathbb{R}^m} \quad & f(p) \\
\text{s.t.} \quad
& p \in \Delta, \\
& W(p,T) \succ 0
\end{aligned}
\tag{P1}
\label{opt:controllability}
\end{equation}
where
\begin{equation}
    f(p) = -\log \det \left(  W(p, T) \right),
\end{equation}
\[
\Delta
\coloneqq
\left\{
p\in\mathbb{R}^m
\;\middle|\;
\sum_{j=1}^m p_j = 1,\;
p_j \geq 0,\quad j=1,\ldots,m
\right\}
\]
is the probability simplex. 
An optimal solution \(\widehat p\) of Problem (P1) is referred to as the volumetric controllability allocation score (V-CAS) for the quantum system, following the controllability scoring framework developed for linear systems \cite{sato2024controllability}. 
Unlike the original linear time-invariant system setting, the present score is constructed from the linearized input-output dynamics of the bilinear Schrödinger system along the nominal trajectory.
We adopt the negative log-determinant criterion to evaluate the collective influence of local Hamiltonian perturbations on the selected observables. 
The output Gramian $W(p,T)$ naturally arises from the linearized input-to-output map and characterizes the output variations achievable under a fixed input-energy budget. 
Since our aim is to evaluate the ability to generate variations across multiple output directions, rather than a large response concentrated in a particular direction, we maximize the volume of the first-order reachable set through $\log\det W(p,T)$. 
Moreover, unlike trace-based criteria such as $\operatorname{tr}(W^{-1})$, the resulting optimal weights are invariant under invertible linear transformations of the output coordinates\cite{sato2026relationship}, making the evaluation independent of the choice of units or basis for the observables. 
The geometric interpretation of the V-CAS is illustrated in Fig.~\ref{fig:controllability_allocation}.

\begin{figure}[t]
    \centering
    \includegraphics[width=\linewidth]{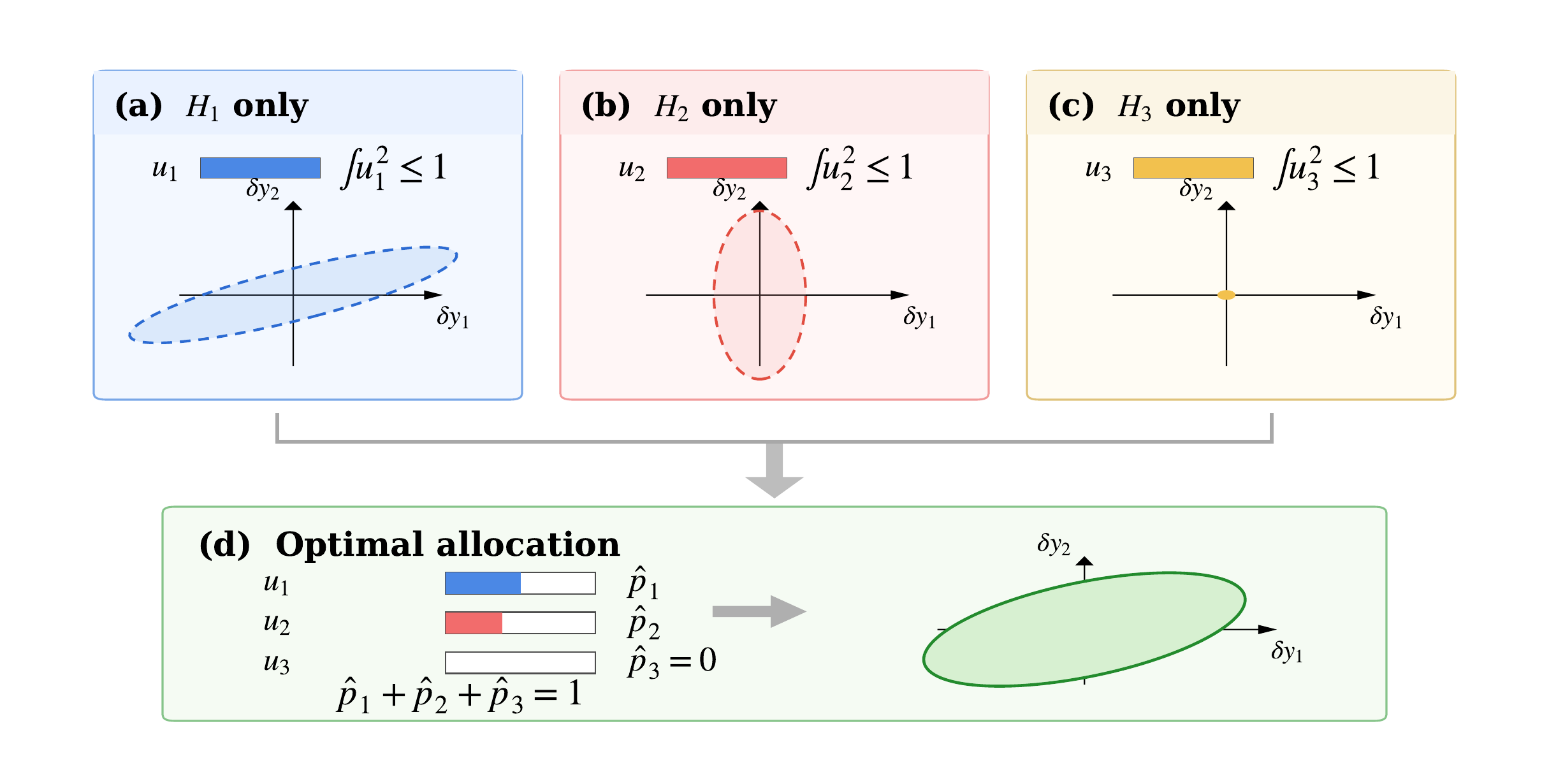}
    \caption{
        Schematic illustration of controllability allocation
        scoring. Panels (a)--(c) show the first-order reachable
        sets in the terminal-output space under a unit
        input-energy budget for the virtual Hamiltonians
        $H_1$, $H_2$, and $H_3$, respectively.
        While $H_1$ and $H_2$ generate output variations
        in different directions, $H_3$ has no influence
        on the selected outputs.
        Panel (d) illustrates the optimal allocation
        $\widehat{p}$, which maximizes the volume of
        the reachable set.
    }
    \label{fig:controllability_allocation}
\end{figure}

For each fixed \(T>0\), Problem \eqref{opt:controllability} has the affine positive-semidefinite Gramian structure studied in \cite{sato2026projected}. Define
\begin{equation}
    \mathcal{W}_T(v):=\sum_{j=1}^{m} v_j W_j(T),
    \qquad
    \mathcal{X}_T:=\left\{
        p\in\mathbb{R}^m:\mathcal{W}_T(p)\succ0
    \right\}.
\end{equation}
Since \(W_j(T)\succeq0\), the feasibility characterization in
\cite[Lemma~3]{sato2026projected} gives
\begin{equation}
    \Delta\cap\mathcal{X}_T\neq\varnothing
    \quad\Longleftrightarrow\quad
    \sum_{j=1}^{m} W_j(T)\succ0.
\end{equation}
Whenever this condition holds, the uniform allocation
\begin{equation}
    p^{\mathrm{uni}}=\frac{1}{m}\mathbf{1}
\end{equation}
is feasible. Existence of an optimal solution to Problem~(P1) then
follows from \cite[Theorem~4]{sato2026projected}.

\subsubsection{Uniqueness}
The uniqueness of the optimal allocation is important for interpreting the V-CAS as a measure of dynamical influence. 
Although the optimal output controllability Gramian is uniquely determined, different allocations may yield the same optimal Gramian. 
Since the V-CAS is defined by the allocation weights themselves, such nonuniqueness may lead to different scores and rankings of local Hamiltonian directions, potentially affecting score-based control-channel selection. 
We therefore examine conditions that guarantee uniqueness of the V-CAS.
Writing \(f(p):=-\log\det W(p,T)\), the gradient formula in
\cite[Lemma~4]{sato2026projected} gives
\begin{equation}
    \frac{\partial f}{\partial p_j}(p)
    =
    -\operatorname{tr}\!\left(
        W(p,T)^{-1}W_j(T)
    \right),
    \qquad j=1,\ldots,m.
    \label{eq:gradient_logdet}
\end{equation}

For the full simplex, the feasible-direction space is
\begin{equation}
    \mathcal{L}_{\Delta}
    =
    \left\{
        v\in\mathbb{R}^m:
        \mathbf{1}^{\top}v=0
    \right\}.
\end{equation}
By \cite[Remark~3]{sato2026projected}, restricted injectivity,
\begin{equation}
    \ker\mathcal{W}_T\cap\mathcal{L}_{\Delta}
    =
    \{0\},
    \label{eq:restricted_injectivity}
\end{equation}
is equivalent to affine independence of
\(W_1(T),\ldots,W_m(T)\), or equivalently to linear independence of
\[
    W_1(T)-W_m(T),\ldots,W_{m-1}(T)-W_m(T).
\]
Whenever Problem~\eqref{opt:controllability} is feasible, this condition guarantees a unique optimal allocation by \cite[Theorem~5]{sato2026projected}.

Since the vector space of real symmetric \(r\times r\) matrices has dimension \(r(r+1)/2\), a necessary condition for affine independence is
\begin{equation}
    m-1\leq \frac{r(r+1)}{2}.
    \label{eq:dimension_condition}
\end{equation}

We next establish a sufficient condition for the affine independence of the Gramian contributions by adapting the argument of \cite[Theorem~2]{sato2025target}.    
\begin{proposition}\label{prop-unique}
For $j=1,\ldots,m$, define the initial output response
\[
q_j:=C(0)[-iH_j\bar{\psi}(0)]\in\mathbb{R}^r.
\]
Assume that the matrices
\[
q_1q_1^\top,\ldots,q_mq_m^\top
\]
are affinely independent.
Then the Gramian contributions $W_1(T),\ldots,W_m(T)$ are affinely independent for all sufficiently small $T>0$ and for almost every $T>0$.
\end{proposition}
\begin{proof}
The proof is almost the same as \cite[Theorem~2]{sato2025target}.

Let $U_{\Delta}$ be a matrix whose columns form an orthonormal basis of $L_{\Delta}$.
By the assumed affine independence, the matrix
\[
K_{\mathrm{inst}}
=
\begin{bmatrix}
\operatorname{vech}(q_1q_1^\top) &
\cdots &
\operatorname{vech}(q_mq_m^\top)
\end{bmatrix}
U_{\Delta}
\]
has full column rank.

Since the channel-wise Gramians satisfy the same short-time limit and real-analyticity properties used in \cite[Theorem~2]{sato2025target}, the argument therein establishes affine independence of $W_1(T),\ldots,W_m(T)$ for all sufficiently small $T>0$ and almost every $T>0$. \hfill$\square$
\end{proof}
Consequently, for almost every $T>0$, the restricted-injectivity condition holds. Hence, whenever Problem \eqref{opt:controllability} is feasible, its optimal solution is unique by \cite[Theorem~5]{sato2026projected}.
\begin{remark}[Independence from the drift Hamiltonian]
The sufficient condition in Proposition \ref{prop-unique} depends only on the initial state \(\psi_0\),  the virtual control Hamiltonians \(H_1,\ldots,H_m\), and the differential of the output map at \(\psi_0\). In particular, it is independent of the drift
Hamiltonian \(H_0\). The exceptional set of terminal times may, however, depend on \(H_0\).
\end{remark}
\begin{remark}[Sufficiency of the initial-response condition]
The condition in Proposition~\ref{prop-unique} is sufficient but not necessary. For example, suppose that
\[
    \bar\psi(0)=\lvert 0\cdots 0\rangle,
    \qquad
    H_j=Z_j .
\]
Since \(\lvert 0\cdots 0\rangle\) is an eigenstate of every \(Z_j\), the
variation
\[
    -\mathrm{i}H_j\bar\psi(0)
\]
is only a global-phase direction. Consequently, for any phase-invariant
output map,
\[
    C(0)\!\left[-\mathrm{i}H_j\bar\psi(0)\right]=0,
    \qquad j=1,\ldots,m,
\]
and the sufficient condition in Proposition ~\ref{prop-unique} is not satisfied.

This failure, however, does not imply that the optimizer is nonunique. Under the drift dynamics generated by \(H_0\), the nominal state and the control sensitivities may evolve into observable directions. Therefore, the finite-time Gramian contributions \(W_j(T)\) may still become affinely independent for \(T>0\), even though their initial responses   vanish.
\end{remark}

\subsection{Numerical Solution of the Allocation Problem}

After computing the channel-wise Gramian contributions \(W_1(T),\ldots,W_m(T)\), we solve Problem \eqref{opt:controllability} using the domain-aware Armijo projected-gradient method of \cite[Algorithm~1]{sato2026projected}. Whenever Problem \eqref{opt:controllability} is feasible, the uniform allocation \(p^{\mathrm{uni}}=m^{-1}\mathbf{1}\) provides a feasible initialization, and convergence of the resulting iteration follows from \cite[Theorem~4]{sato2026projected}. The numerical construction of the Gramian contributions, including the MPS/TEBD implementation for many-body systems, is described in the next section.

\section{MPS/TEBD Implementation for Quantum Many-Body Systems}\label{sec:MPS}

Section~\ref{sec:quantum-controllability} reduced the computation of the volumetric controllability allocation score (V-CAS) to the construction of the channel-wise Gramian contributions $W_j(T)$.
In particular, the adjoint representation shows that these matrices can be evaluated from the nominal trajectory $\bar\psi(t)$, the adjoint states $z_\ell(T,t)$, and their inner products with the virtual Hamiltonian directions $(-\mathrm{i}H_j)\bar\psi(t)$, without explicitly constructing the state-transition operator.

Although each output Gramian contribution $W_j(T)$ is only an $r\times r$ matrix, the states appearing in this representation belong to the full many-body Hilbert space. For an $n$-qubit system, this space has dimension $2^n$, so direct propagation of the nominal and adjoint states becomes computationally prohibitive as the system size increases.

To address this computational difficulty, we exploit the tensor-product structure of one-dimensional local quantum many-body systems and represent the required forward and backward states by matrix product states (MPS). 
MPS includes the product-state ansatz commonly used in site-factorized mean-field approximations of lattice models, while also allowing entanglement between sites.
Combined with the time-evolving block decimation (TEBD) algorithm, this allows the adjoint representation derived in Section~\ref{sec:quantum-controllability} to be evaluated using local tensor updates and MPS contractions, without explicitly forming full many-body state vectors or state-transition matrices.

The assumptions introduced below are required only for the scalable
MPS/TEBD implementation in this section and are not needed for the
formulation and theoretical results in Sections~\ref{sec:Prelim} and~\ref{sec:quantum-controllability}.

We consider a one-dimensional chain of $n$ quantum sites and employ the
MPS representation for the corresponding many-body states
\cite{vidal2004efficient}. The sites are arranged in a linear order and
labeled by a single index $i=1,\ldots,n$. We assume that the Hamiltonian
is local with respect to this ordering, namely, that it consists of
on-site terms and interaction terms supported within a fixed finite
range along the chain. This locality structure is the computational
setting in which MPS-based time evolution can be efficiently applied.

Although this setting is restricted to one-dimensional local systems,
such models constitute an important class of physically motivated
quantum many-body systems. Examples include quantum spin chains and
one-dimensional lattice models, and related effective descriptions arise
in quantum-simulation platforms such as ultracold atoms in optical
lattices and trapped-ion quantum simulators
\cite{bloch2008many,monroe2021programmable}. In such platforms, available
controls and physically relevant perturbations are often local or
few-body, reflecting the spatial structure of the underlying system.
Consequently, identifying which sites have a large dynamical influence
on the global evolution and on the selected output observables is
relevant to practical tasks such as control-input placement, sensitivity
analysis with respect to local disturbances, and the design of quantum
simulation experiments.

For such systems, MPS provides a natural computational representation
when the entanglement growth over the time horizon of interest remains
moderate. In this regime, the nominal trajectory and the adjoint states
required for the channel-wise Gramian construction can be approximated
efficiently in MPS form
\cite{paeckel2019time,orus2014practical,vidal2004efficient}.

\subsection{Matrix Product States: General Scope and Computational Advantage}

MPS representations \cite{vidal2004efficient} are particularly effective for one-dimensional quantum systems whose Hamiltonians are local with respect to the chain structure. Under such conditions, the time-evolved quantum state can often be approximated accurately with a moderate bond dimension \(\chi\) over finite time horizons. The bond dimension \(\chi\) is an internal rank parameter of the MPS representation. For an \(n\)-qubit system, it satisfies $1\leq \chi\leq 2^{\lfloor n/2 \rfloor}$. 
Taking \(\chi = 2^{\lfloor n/2 \rfloor}\) allows an exact representation of an arbitrary state, whereas smaller values of \(\chi\) correspond to low-rank approximations of the state. Thus, when the relevant states can be approximated accurately with moderate \(\chi\), the MPS representation reduces both storage and computational costs compared with the full state vector representation.

The MPS framework efficiently supports the operations required in our
computation:
\begin{itemize}
    \item time evolution under local Hamiltonians,
    \item application of local operators,
    \item inner-product evaluation between pure quantum states.
\end{itemize}
These operations can be implemented through local tensor updates and sequential tensor contractions. For a chain of length \(n\), the typical computational cost scales polynomially in the bond dimension. In standard MPS time-evolution algorithms, the short-time propagator is decomposed into local operators that are applied successively along the chain, with a typical cost of order \(\mathcal{O}(n\chi^3)\) per time step \cite{paeckel2019time}.
The storage cost of an MPS with bond dimension  \(\chi\) scales as \(\mathcal{O}(n\chi^2)\), whereas storing a full many-body state vector requires \(\mathcal{O}(2^n)\) complex numbers \cite{orus2014practical}.

The approximation accuracy is controlled by the truncation threshold used
in the singular value decompositions. For finite time intervals, the
cumulative truncation error can be monitored, and the accuracy can be
systematically improved by increasing the bond dimension \(\chi\).

\subsection{MPS Evaluation of the Output Sensitivities}

Recall from Section~\ref{sec:quantum-controllability} that the channel-wise output sensitivity admits
the adjoint representation
\begin{equation}
    [y_j(T,t)]_\ell
    =
    2\operatorname{Re}
    \left\langle
        z_\ell(T,t),
        (-\mathrm{i}H_j)\bar\psi(t)
    \right\rangle ,
    \label{eq:mps_output_sensitivity_component}
\end{equation}
where $z_\ell(T,t)$ is the adjoint state defined in Section~\ref{sec:quantum-controllability} as
\begin{equation}
    z_\ell(T,t)
    =
    \Phi(T,t)^*O_\ell\bar\psi(T).
\end{equation}

Throughout this section, the terminal time $T$ is fixed. For notational
simplicity, we write
\begin{equation}
    y_j(t):=y_j(T,t),
    \qquad
    z_\ell(t):=z_\ell(T,t).
\end{equation}
Hence, the construction of the channel-wise Gramian does not require
explicit formation of the state-transition operator. Instead, it is
sufficient to compute the nominal state $\bar\psi(t)$, the adjoint
states $z_\ell(t)$, and their inner products with the locally perturbed
states $(-\mathrm{i}H_j)\bar\psi(t)$. These quantities can be
represented and evaluated directly in MPS form.

We now describe their evaluation at the quadrature points $t_k$.
First, the nominal trajectory $\bar\psi(t_k)$ is computed by a standard
MPS time-evolution algorithm and stored in MPS form at the quadrature
points $t_k$.

At the terminal time $T$, the terminal value of the adjoint state
associated with each observable $O_\ell$ is
\begin{equation}
    z_\ell(T)
    =
    z_\ell(T,T)
    =
    O_\ell\bar\psi(T).
    \label{eq:mps_terminal_observable_state}
\end{equation}
We assume that each observable $O_\ell$ is local, few-body, or admits
an MPO representation with moderate bond dimension. Therefore, the
state $z_\ell(T)=O_\ell\bar\psi(T)$ can be efficiently represented as
an MPS.

For each time point $t_k$, we evaluate the adjoint state by backward
propagation as
\begin{equation}
    z_\ell(t_k)
    =
    \Phi(T,t_k)^\ast z_\ell(T).
    \label{eq:mps_backward_observable_state}
\end{equation}

Next, for each control direction $H_j$ and each time point $t_k$, we
form the perturbation state
\begin{equation}
    r_j(t_k)
    :=
    (-\mathrm{i}H_j)\bar\psi(t_k).
    \label{eq:mps_local_perturbation_state}
\end{equation}
Since $H_j$ is a local operator, this state can also be efficiently
constructed in MPS form.

Using the adjoint representation derived in Section~\ref{sec:quantum-controllability}, the output
sensitivity at each quadrature point is then computed as
\begin{equation}
    \bigl(y_j(t_k)\bigr)_\ell
    =
    2\operatorname{Re}
    \left\langle
        z_\ell(t_k),
        r_j(t_k)
    \right\rangle .
    \label{eq:mps_output_sensitivity_mps}
\end{equation}

Thus, the computation of $y_j(t_k)$ requires only operations that are
compatible with the MPS framework: application of local or few-body
operators, adjoint time evolution, and MPS inner-product contractions.

\subsection{Practical Construction of the Output Gramian}\label{practical-construction}

Using the shorthand introduced above, the channel-wise Gramian
contribution defined in Section~\ref{sec:quantum-controllability} can be written as
\begin{equation}
    W_j(T)
    =
    \int_0^T
    y_j(t)y_j(t)^\top \,\mathrm dt .
    \label{eq:mps_Wmu_output}
\end{equation}

In numerical implementation, we introduce the discrete Gramian
\begin{equation}
W_j^h(T)
:=
h\sum_{k=0}^{N}
\omega_k
y_j(t_k)y_j(t_k)^\top,
\label{eq:exact-discrete-gramian}
\end{equation}
where $\omega_k\geq 0$ are quadrature weights. For example, for the
trapezoidal rule,

\[
\omega_0=\omega_N=\frac{1}{2},
\qquad
\omega_k=1,
\quad k=1,\ldots,N-1.
\]
The MPS/TEBD computation produces approximate sensitivity vectors
$\widehat y_j(t_k)$ and the corresponding approximate discrete Gramian
\begin{equation}
\widehat W_j^h(T)
:=
h\sum_{k=0}^{N}
\omega_k
\widehat y_j(t_k)\widehat y_j(t_k)^\top.
\label{eq:approx-discrete-gramian}
\end{equation}
\subsection{Error Bound for the MPS/TEBD Approximation}
\label{subsec:mps-error}

We next quantify the approximation error introduced by the
Trotter decomposition and the MPS truncation in the computation
of the output Gramian.

For the error analysis, let
\[
    \bar\psi_k:=\bar\psi(t_k),
    \qquad
    z_{k,l}:=z_l(t_k),
    \qquad
    y_{j,k}:=y_j(t_k)
\]
denote the exact quantities at $t_k$, and let
\[
    \widehat\psi_k,
    \qquad
    \widehat z_{k,l},
    \qquad
    \widehat y_{j,k}
\]
denote their MPS/TEBD approximations, where
\begin{equation}
    (\widehat y_{j,k})_l
    :=
    2\operatorname{Re}
    \left\langle
        \widehat z_{k,l},
        -\mathrm{i}H_j\widehat\psi_k
    \right\rangle .
    \label{eq:approx-sensitivity-main}
\end{equation}

We assume a uniform time step $h$ and second-order symmetric
Trotter evolution.  A standard second-order product-formula
estimate gives a constant $C_{\mathrm{Tr}}>0$, independent of
$h$, such that
\begin{equation}
    \|U-S\|_2
    \leq
    C_{\mathrm{Tr}}h^3,
    \label{eq:trotter-local-main}
\end{equation}
where $U$ and $S$ denote one exact and one Trotterized time step,
respectively.  An explicit commutator expression for
$C_{\mathrm{Tr}}$ is given in Appendix~\ref{app:mps-error}.

Let $B_q^\psi$ denote the number of Schmidt truncations
during the forward TEBD step $q$, and let
$\delta_{q,b}^{\psi}$, $b=1,\ldots,B_q^\psi$, be the discarded
weight at the $b$-th truncation.
For the backward calculation, let $\delta_{q,l,b}^{z}$,
$b=1,\ldots,B_{q,l}^{z}$, and
$\delta_{N,l,b}^{\mathrm{term}}$,
$b=1,\ldots,B_{N,l}^{\mathrm{term}}$, denote the corresponding
discarded weights during backward propagation and terminal
compression, respectively.

We define
\begin{equation}
    E_k^{\psi,\mathrm{Tr}}
    :=
    t_k C_{\mathrm{Tr}}h^2, \quad
    E_k^{\psi,\mathrm{cut}}
    :=
    \sum_{q=1}^{k}
    \sum_{b=1}^{B_q^\psi}
    \sqrt{2\delta_{q,b}^{\psi}},
    \label{eq:error}
\end{equation}
and, with $L_l:=\|O_l\|_2$,
\begin{align}
    E_{k,l}^{z,\mathrm{Tr}}
    &:=
    L_l(2T-t_k)C_{\mathrm{Tr}}h^2,
    \label{eq:backward-trotter-budget-main}\\
    E_{k,l}^{z,\mathrm{cut}}
    &:=
    L_lE_N^{\psi,\mathrm{cut}}
    +
    \sum_{b=1}^{B_{N,l}^{\mathrm{term}}}
    \sqrt{\delta_{N,l,b}^{\mathrm{term}}}
    +
    \sum_{q=k}^{N-1}
    \sum_{b=1}^{B_{q,l}^{z}}
    \sqrt{\delta_{q,l,b}^{z}}.
    \label{eq:backward-cut-budget-main}
\end{align}
Then the forward and backward state errors satisfy
\begin{equation}
    \|\bar\psi_k-\widehat\psi_k\|_2
    \leq
    E_k^{\psi,\mathrm{Tr}}
    +
    E_k^{\psi,\mathrm{cut}},\quad
    \|z_{k,l}-\widehat z_{k,l}\|_2
    \leq
    E_{k,l}^{z,\mathrm{Tr}}
    +
    E_{k,l}^{z,\mathrm{cut}}.
    \label{eq:backward-state-error-main}
\end{equation}
The derivation of these bounds is given in
Appendix~\ref{app:mps-error}.
The truncation terms are directly computable from the discarded
Schmidt coefficients recorded during the TEBD calculation, whereas
the Trotter terms follow from standard product-formula error
estimates.

We next propagate these bounds to the output sensitivities.
Let $a_j:=\|H_j\|_2.$ Define
\begin{equation}
    \eta_{j,k,l}^{\mathrm{Tr}}
    :=
    2a_j
    \left(
        L_lE_k^{\psi,\mathrm{Tr}}
        +
        E_{k,l}^{z,\mathrm{Tr}}
    \right), \quad
    \eta_{j,k,l}^{\mathrm{cut}}
    :=
    2a_j
    \left(
        L_lE_k^{\psi,\mathrm{cut}}
        +
        E_{k,l}^{z,\mathrm{cut}}
    \right),
    \label{eq:error-eta}
\end{equation}

and
\begin{align}
    \eta_{j,k}^{\mathrm{Tr}}
    &:=
    \left(
        \sum_{l=1}^{r}
        \left(\eta_{j,k,l}^{\mathrm{Tr}}\right)^2
    \right)^{1/2}, \quad
    \eta_{j,k}^{\mathrm{cut}}
    :=
    \left(
        \sum_{l=1}^{r}
        \left(\eta_{j,k,l}^{\mathrm{cut}}\right)^2
    \right)^{1/2}.
\end{align}
Then
\begin{equation}
    \|y_{j,k}-\widehat y_{j,k}\|_2
    \leq
    \eta_{j,k}^{\mathrm{Tr}}
    +
    \eta_{j,k}^{\mathrm{cut}}.
    \label{eq:sensitivity-error-main}
\end{equation}

Let $W_j^h(T)$ and $\widehat W_j^h(T)$ denote the discrete Gramians constructed from $y_{j,k}$ and $\widehat y_{j,k}$, respectively, using the same numerical integration rule as in Section~\ref{practical-construction}.

\begin{theorem}[Error of the MPS/TEBD Gramian approximation]
\label{thm:mps-gramian-error}
The discrete-Gramian error satisfies
\begin{equation}
    \|W_j^h(T)-\widehat W_j^h(T)\|_2
    \leq
    E_j^{\mathrm{Tr}}
    +
    E_j^{\mathrm{cut}}
    +
    E_j^{\mathrm{cross}},
    \label{eq:gramian-error-decomposition-main}
\end{equation}
where
\begin{align}
    E_j^{\mathrm{Tr}}
    &:=h
    \sum_k \omega_k
    \left[
        2\|\widehat y_{j,k}\|_2
        \eta_{j,k}^{\mathrm{Tr}}
        +
        \left(\eta_{j,k}^{\mathrm{Tr}}\right)^2
    \right],
    \label{eq:gramian-trotter-main}\\
    E_j^{\mathrm{cut}}
    &:=h
    \sum_k \omega_k
    \left[
        2\|\widehat y_{j,k}\|_2
        \eta_{j,k}^{\mathrm{cut}}
        +
        \left(\eta_{j,k}^{\mathrm{cut}}\right)^2
    \right],
    \label{eq:gramian-cut-main}\\
    E_j^{\mathrm{cross}}
    &:=2h
    \sum_k \omega_k
    \eta_{j,k}^{\mathrm{Tr}}
    \eta_{j,k}^{\mathrm{cut}},
    \label{eq:gramian-cross-main}
\end{align}
and $\omega_k\geq0$ are the weights used in the numerical integration
of the Gramian.
\end{theorem}

The proof is given in Appendix~\ref{app:mps-error}.
In particular, the second-order Trotter contribution is
$O(h^2)$ over a fixed time horizon, whereas the MPS-truncation
contribution is determined a posteriori by the Schmidt
coefficients discarded during the actual TEBD calculation.
Thus, the bound separates the two main numerical error sources
and quantifies their propagation to the discrete output Gramian.
The quadrature error between the continuous-time Gramian and the
discrete Gramian is not considered here.
\subsection{Summary of the Algorithm}
The complete procedure is summarized in Algorithm \ref{alg:mps-controllability-scoring}.

\begin{algorithm}[t]
    \caption{MPS-Based Controllability Allocation Scoring}
    \label{alg:mps-controllability-scoring}
    \textbf{Input:} Drift Hamiltonian $H_0$, initial state $\bar{\psi}(0)$,\\
    final time $T$, virtual inputs $H_1,H_2,\ldots,H_m$,
    observables $O_1,\ldots,O_r$,\\
    quadrature points $t_0,\ldots,t_N$, time step $h$,
    and quadrature weights $\omega_0,\ldots,\omega_N$.\\
    \textbf{Output:} Numerical approximation $\widehat p$ of the V-CAS
    \begin{algorithmic}[1]
        \STATE Compute the approximate nominal trajectory
        $\widehat{\psi}_k$ at $t_k$ in MPS form by forward TEBD

        \FOR{$\ell=1,2,\ldots,r$}
            \STATE Form $O_\ell\widehat{\psi}_N$ in MPS form and,
            if necessary, compress it to obtain $\widehat z_{N,\ell}$
            \STATE Compute $\widehat z_{k,\ell}$ for
            $k=N-1,\ldots,0$ by backward TEBD
        \ENDFOR

        \FOR{$j=1,2,\ldots,m$}
            \STATE Set $\widehat W_j^h(T)=0$
            \FOR{each quadrature point $t_k$}
                \STATE Compute
                \[
                    \widehat r_{j,k}
                    =
                    (-\mathrm{i}H_j)\widehat{\psi}_k
                \]
                \STATE Compute
                \[
                    \widehat y_{j,k}
                    =
                    \begin{bmatrix}
                        2\operatorname{Re}
                        \langle \widehat z_{k,1},\widehat r_{j,k}\rangle\\
                        \vdots\\
                        2\operatorname{Re}
                        \langle \widehat z_{k,r},\widehat r_{j,k}\rangle
                    \end{bmatrix}
                \]
                \STATE
                $\widehat W_j^h(T)
                \leftarrow
                \widehat W_j^h(T)
                +
                h\omega_k
                \widehat y_{j,k}\widehat y_{j,k}^{\top}$
            \ENDFOR
        \ENDFOR

        \STATE Solve the numerical approximation of
        Problem~\eqref{opt:controllability}, with
        $W_j(T)$ replaced by $\widehat W_j^h(T)$,
        using the method of \cite[Algorithm~1]{sato2026projected}

        \RETURN $\widehat p$
    \end{algorithmic}
\end{algorithm}

\section{Numerical Experiments}
\label{sec:numerical}

We conduct three numerical studies. 
First, we compare the proposed V-CAS with static quantum centralities.
Second, we examine whether the finite-time linearized response captured by the score remains informative for nonlinear terminal-output control in hard-core-boson chains. 
Finally, we evaluate the accuracy and computational cost of the MPS implementation.

\subsection{Comparison with Static Quantum Centralities}
\label{subsec:compare}

We compare the proposed V-CAS with mutual-information strength \cite{sundar2018complex}, concurrence strength \cite{nokkala2024complex}, and mutual-information betweenness centrality \cite{sundar2018complex} for four eight-qubit models: central spin, small world, modular two-cluster, and disordered XXZ.
For all systems, we use
\[
    H_j=O_j=Z_j,\qquad T=0.6,\qquad \Delta t=0.02.
\]
The static centralities are evaluated from the uncontrolled terminal state
$\bar\psi(T)$. Their definitions and the detailed model parameters are given
in Appendix~\ref{sec:numerical_experiment_settings}.

\begin{table}[t]
\centering
\caption{
Pearson correlation / Spearman rank correlation between the proposed V-CAS and static quantum centralities.
}
\label{tab:static-centrality-correlation}
\begin{tabular}{lccc}
\hline
System
& Mutual information
& Concurrence
& MI betweenness
\\
\hline
Central spin
& $0.9306/0.8333$
& $0.7950/0.7857$
& $0.7190/0.5774$
\\
Small world
& $0.5352/0.7857$
& $0.5486/0.7381$
& $-0.3449/-0.3130$
\\
Modular two-cluster
& $0.6165/0.7619$
& $0.7198/0.9762$
& $0.2131/0.4665$
\\
Disordered XXZ
& $0.1545/0.0714$
& $0.2398/0.1190$
& $-0.4598/-0.4470$
\\
\hline
\end{tabular}
\end{table}
Table~\ref{tab:static-centrality-correlation} shows that the agreement between the proposed score and the static centralities varies across systems and measures. 
While substantial agreement is observed in some cases, the weak or negative correlations in the disordered XXZ chain show that the proposed V-CAS can induce a different ranking of influential sites. 
This supports its interpretation as a finite-time dynamical influence measure rather than a static correlation- or entanglement-based centrality.

\subsection{Relation to Nonlinear Output Control}\label{subsec:control_experiment}

We next examine whether the finite-time linearized response captured by the proposed V-CAS remains informative for nonlinear terminal-output control. 
For each system, we generate terminal targets, select two control channels from all available local controls, and optimize those channels to approach the target as closely as possible. 
We then compare the proposed score-based selection with several alternative placement criteria.

For an uncontrolled output $y_0(T)$ and target $y_{\rm tar}$, performance is measured by the normalized terminal error
\[
    e_{\rm norm}
    =
    \frac{
        \|y(T)-y_{\rm tar}\|_2
    }{
        \|y_0(T)-y_{\rm tar}\|_2
    }.
\]
The score-based top-two rule selects the two largest components of the continuous V-CAS vector and is not assumed to solve a discrete channel-selection problem.
Details of target generation and nonlinear optimization are given in Appendix~\ref{app:hcb_benchmark}.

We consider a half-filled hard-core-boson chain of length $L=10$, described by a standard Hamiltonian with nearest-neighbor hopping, density interactions, and site-dependent potentials \cite{cazalilla2011one}:
\begin{equation}
    H_0
    =
    -\sum_{j=1}^{L-1}
    t_j
    \left(
        b_j^\dagger b_{j+1}
        +
        b_{j+1}^\dagger b_j
    \right)
    +
    V\sum_{j=1}^{L-1}n_jn_{j+1}
    +
    \sum_{j=1}^{L}h_jn_j .
    \label{eq:hcb_hamiltonian}
\end{equation}
Here, $n_j=b_j^\dagger b_j$, and the hard-core constraint allows at most one particle per site.
Since $H_0$ conserves the total particle number, the dynamics at half filling with five particles remain in a finite-dimensional Hilbert space of dimension $\binom{10}{5}=252$.
We set $V=1$ and consider two representative Hamiltonian configurations: a central two-site barrier and harmonic confinement.
The virtual and physical control Hamiltonians are the local density potentials \(H_j=n_j\), and the initial state is the charge-density-wave state \(\bar\psi(0)=\lvert1010101010\rangle\). 
The outputs span the nine independent density-profile directions. 
Detailed model parameters and output coordinates are given in Appendix~\ref{app:hcb_benchmark}.

For each system instance, the V-CAS is computed once along the uncontrolled trajectory over the horizon \(T=2\), and the same allocation is subsequently used for all terminal targets.
The two largest components of the V-CAS define the score-based two-site preselection.
We compare this selection with discrete log-determinant optimization, trace-based selection, a geometry-dependent physical heuristic, the boundary pair, and three fixed random pairs.
All placement rules are target independent and are determined once for each system before the nonlinear-control experiments.

Figure~\ref{fig:hcb_cas_barrier_trap} illustrates the spatial structure of the two Hamiltonian configurations and the corresponding V-CAS allocations.
In both configurations, we consider controlling the terminal particle-density profile by varying local on-site potentials.
In Fig.~\ref{fig:hcb_cas_barrier_trap}(a), a potential barrier is introduced at the two central sites, restricting particle transport across the center of the chain.
The score-based top-two rule selects sites $\{3,6\}$, suggesting that controlling the on-site potentials at the barrier and in its surrounding region is effective for manipulating the terminal density profile.
In contrast, Fig.~\ref{fig:hcb_cas_barrier_trap}(b) considers a harmonic confining potential that increases from the center toward both ends of the chain.
In this configuration, the score-based top-two rule selects sites $\{3,8\}$, suggesting that controlling the on-site potentials on opposite sides of the confined region is effective.
These allocations are determined from the uncontrolled dynamics, independently of the terminal targets, and demonstrate that the preferred control locations depend on the spatial structure of the underlying potential.

\begin{figure}[t]
    \centering
    \includegraphics[width=\linewidth]{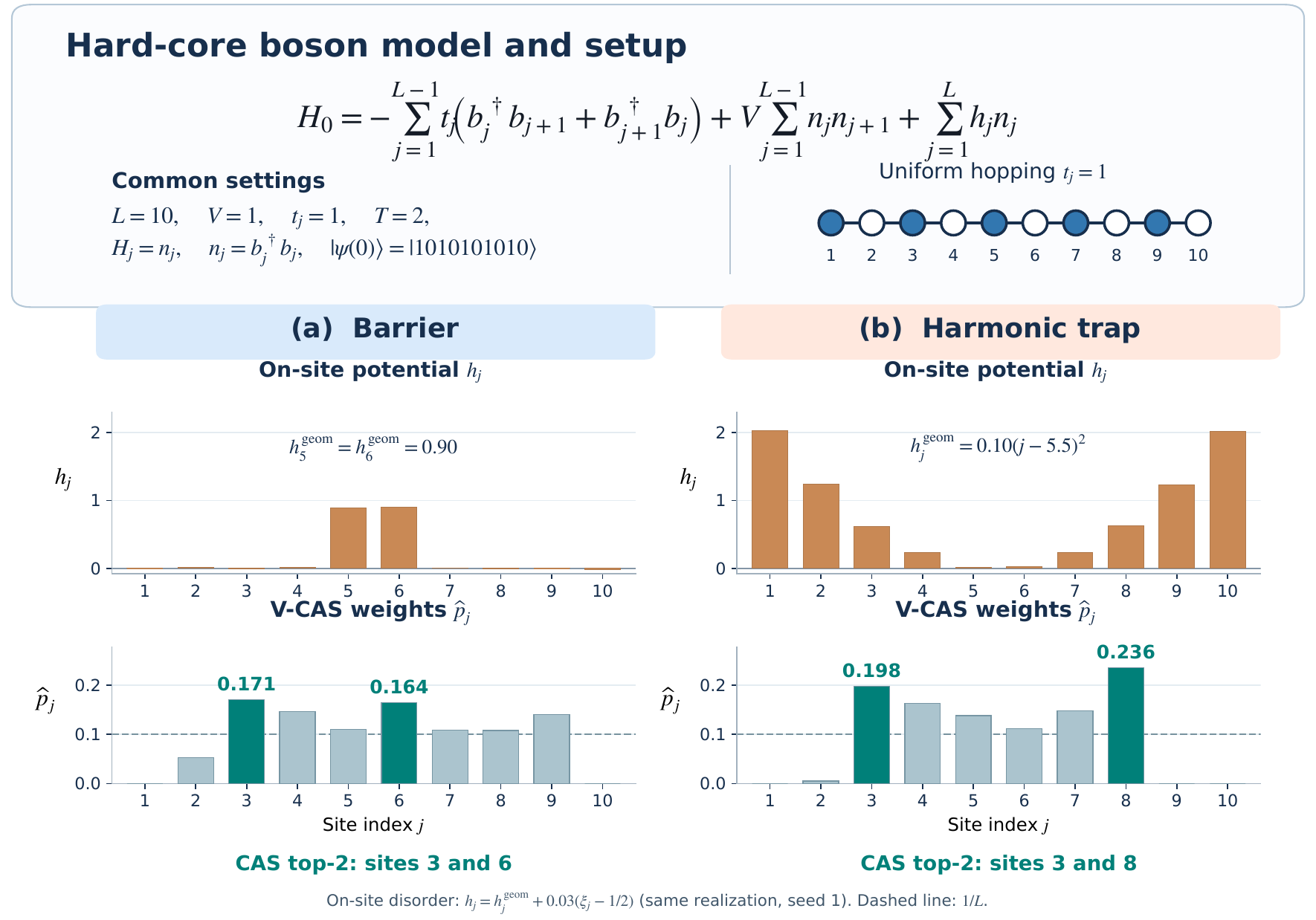}
    \caption{%
    Spatial structure and controllability-based input allocation in two
    half-filled hard-core-boson chains.
    The top panel shows the drift Hamiltonian, common parameters, and the
    initial charge-density-wave state.
    The on-site potentials are
    $h_j=h_j^{\mathrm{geom}}+0.03(\xi_j-1/2)$, with the same disorder
    realization (seed $1$) in both configurations.
    (a) Central two-site barrier with
    $h_5^{\mathrm{geom}}=h_6^{\mathrm{geom}}=0.90$.
    (b) Harmonic confinement with
    $h_j^{\mathrm{geom}}=0.10(j-5.5)^2$.
    In each case, the lower panel shows the V-CAS weights
    $\widehat p_j$ computed from the uncontrolled trajectory at $T=2$
    using MPS/TEBD. Dark-green bars indicate the two sites selected by
    the score-based top-two rule: $\{3,6\}$ for the barrier and
    $\{3,8\}$ for the trap. Dashed lines mark the uniform
    allocation $p_j=1/L$.
    }
    \label{fig:hcb_cas_barrier_trap}
\end{figure}

We generate five target directions from global, center-weighted, edge-weighted, staggered, and tilt teacher-control families.
For each family, one teacher-control direction is generated using all $L$ local density controls, and the same direction is scaled to amplitudes $0.15$, $0.30$, and $0.50$.
The resulting nonlinear trajectories define $15$ target--amplitude cases for each geometry and $30$ cases over the two geometries.

For each target, we solve the nonlinear terminal-output control problem after restricting the available controls to the selected two-site set.
Only the distinct site pairs returned by the target-independent placement rules described above are evaluated; when multiple rules select the same pair, the corresponding optimization is performed only once.
The controlled dynamics and the adjoint calculation are evaluated using the MPS/TEBD implementation described in Sec.~\ref{sec:MPS}.
Further details of the target generation and nonlinear-control optimization are provided in Appendix~\ref{app:hcb_benchmark}.

\textit{Finite-time uniqueness.}
Since the charge-density-wave initial state is an eigenstate of every \(n_j\), the initial-response condition of Proposition~\ref{prop-unique} is not satisfied.
Nevertheless, feasibility and affine independence of the finite-time Gramian contributions at \(T=2\) are numerically confirmed for all system instances; see Appendix~\ref{app:hcb_uniqueness} for the complete diagnostics.

\begin{table}[t]
\centering
\caption{
Mean normalized terminal errors for the target-independent two-site
placement rules in the hard-core-boson benchmark.
Each geometry contains $15$ target--amplitude cases.
The overall column is the mean over all $30$ cases, and ``Wins'' denotes
the number of cases in which a method attains the smallest terminal error
among the compared placement rules.
}
\label{tab:hcb_aggregate}
\small
\setlength{\tabcolsep}{4pt}
\begin{tabular}{lcccc}
\toprule
Selection rule
& Barrier
& Trap
& Overall
& Wins \\
\midrule
Score top-2
& \(\mathbf{0.417}\)
& \(0.561\)
& \(\mathbf{0.489}\)
& \(\mathbf{6}\) \\
Discrete log-det
& \(0.533\)
& \(0.720\)
& \(0.626\)
& \(4\) \\
Trace
& \(0.661\)
& \(\mathbf{0.413}\)
& \(0.537\)
& \(\mathbf{6}\) \\
Physical heuristic
& \(0.705\)
& \(0.682\)
& \(0.693\)
& \(1\) \\
Boundary
& \(0.531\)
& \(0.700\)
& \(0.615\)
& \(5\) \\
Random fixed 1
& \(0.600\)
& \(0.923\)
& \(0.761\)
& \(2\) \\
Random fixed 2
& \(0.717\)
& \(0.679\)
& \(0.698\)
& \(2\) \\
Random fixed 3
& \(0.664\)
& \(0.616\)
& \(0.640\)
& \(4\) \\
\bottomrule
\end{tabular}
\end{table}
The results in Table~\ref{tab:hcb_aggregate} support the proposed V-CAS as a target-independent, output-aware preselection criterion rather than as a universally optimal discrete placement rule.
The score-based top-two selection achieves the lowest aggregate mean normalized terminal error among the compared target-independent placement rules.
Moreover, its mean error is lower than that of the exhaustive discrete log-determinant placement in both configurations. 
The score-based rule attains the smallest overall mean normalized terminal error, while its number of wins is tied with the trace-based rule.
Although the trace-based rule achieves a smaller mean error for the trap configuration, the score-based selection provides the smallest aggregate mean error
across the two configurations.This configuration dependence is consistent with the finite-time, target-independent nature of the score, while the top-two rule should be interpreted only as a sparse truncation of the continuous allocation.

\subsection{Accuracy and Computational Cost of the MPS Approximation}
\label{sec:mps_benchmark}

Finally, we evaluate the accuracy and computational cost of the MPS-based V-CAS computation by comparison with an exact full-Hilbert-space implementation.
We consider a disordered nearest-neighbor transverse-field Ising chain
\begin{equation}
    H_0
    =
    \sum_{j=1}^{n}
    \left(
        0.8 X_j+h_j^z Z_j
    \right)
    +
    \sum_{j=1}^{n-1}
    J_j Z_j Z_{j+1},
\end{equation}
with the initial state
\[
    \bar{\psi}(0)=\ket{0}^{\otimes n}.
\]
The local virtual Hamiltonians and output observables are chosen as
\[
    H_j=O_j=Z_j,
    \qquad
    j=1,\ldots,n.
\]
We use the terminal time \(T=0.6\) and time step \(\Delta t=0.02\).
For the MPS calculation, the maximum bond dimension is \(\chi_{\max}=10\), with singular-value cutoff \(10^{-10}\).
The detailed disorder realization and computational settings are given in Appendix~\ref{app:mps_benchmark}.

\begin{table}[t]
    \centering
    \caption{
        Accuracy and computational speedup of the MPS implementation
        relative to the exact full-Hilbert-space computation.
        The speedup is defined as the exact runtime divided by the
        MPS runtime.
    }
    \label{tab:mps_accuracy}
    \begin{tabular}{cccccc}
        \hline
        \(n\)
        & Pearson
        & Cosine
        & \(\ell_1\) error
        & Top-3 overlap
        & Speedup
        \\
        \hline
        6
        & 0.9999
        & 1.0000
        & \(3.93\times10^{-8}\)
        & 1.000
        & \(0.08\times\)
        \\
        8
        & 0.9819
        & 1.0000
        & \(7.27\times10^{-9}\)
        & 1.000
        & \(0.35\times\)
        \\
        10
        & 0.9995
        & 1.0000
        & \(9.71\times10^{-8}\)
        & 1.000
        & \(2.38\times\)
        \\
        \hline
    \end{tabular}
\end{table}
\begin{figure}[t]
    \centering
    \includegraphics[width=0.85\linewidth]{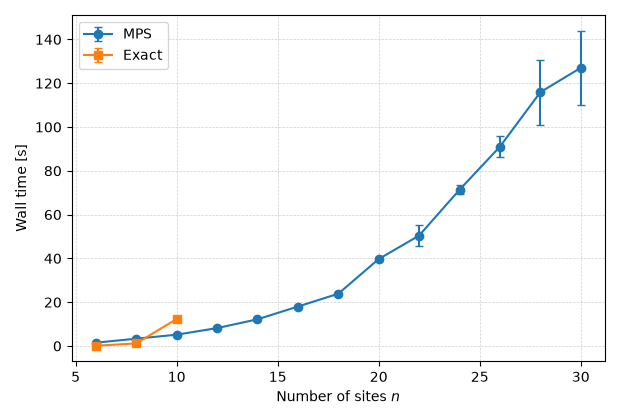}
    \caption{
    Runtime comparison between the MPS and exact full-Hilbert-space
    implementations for the disordered nearest-neighbor transverse-field
    Ising chain.
    Each point shows the mean of three runs, and the error bars indicate
    one standard deviation.
    }
    \label{fig:mps_runtime}
\end{figure}

Table~\ref{tab:mps_accuracy} compares the V-CAS vectors obtained from the MPS and exact implementations on the common system sizes \(n=6,8,10\).
For all three sizes, the cosine similarity is \(1.0000\), the \(\ell_1\) distance remains below \(10^{-7}\), and the three highest-scoring sites coincide exactly.
The Pearson correlation is also high, ranging from \(0.9819\) to \(0.9999\).
These results show that the MPS implementation accurately reproduces both the score distribution and the highest-ranked local control directions in the examined regime.

Figure~\ref{fig:mps_runtime} compares the computational runtimes.
For small systems, the tensor-network overhead makes the MPS calculation slower than the exact full-Hilbert-space implementation.
At \(n=10\), however, the MPS calculation is approximately \(2.38\) times faster.
The cost of the exact computation increases rapidly with system size, whereas the MPS implementation remains applicable up to the largest tested system, \(n=30\).
Thus, in the examined regime, the MPS formulation accurately reproduces the V-CAS while extending the practically accessible system size.
\section{Concluding Remarks}\label{sec:conclusion}
\subsection{Summary}

This paper developed a controllability allocation scoring framework for finite-dimensional quantum systems.

We first formulated the linear input-to-terminal-output operator associated with the linearized Schrödinger dynamics and defined the corresponding output controllability Gramian. 
The Gramian is positive semidefinite, its range coincides with the reachable subspace of terminal-output variations, and its positive definiteness is equivalent to output controllability of the linearized input-output system. 
We also established the associated minimum-energy characterization: the Moore--Penrose pseudoinverse of the Gramian determines the minimum input energy required to generate a prescribed reachable terminal-output variation.

Based on this Gramian, we introduced virtual local control Hamiltonians and defined a V-CAS through the simplex-constrained log-determinant optimization problem \eqref{opt:controllability}.
The resulting weight vector describes an optimal allocation of perturbation directions that maximizes the volume of the first-order reachable set in the selected output space. 
Thus, the proposed score should be interpreted as an output-dependent and trajectory-dependent measure of the cooperative dynamical influence of the virtual control directions.

We further derived a sufficient condition for uniqueness of the V-CAS. 
The condition is expressed solely in terms of the initial state, the virtual control Hamiltonians, and the differential of the output map at the initial state. 
In particular, the condition does not depend on the drift Hamiltonian. 
Under this condition, the objective is strictly convex along the feasible directions of the simplex.
To make the framework applicable to quantum many-body systems, we formulated the Gramian directly in a prescribed observable space. 
This avoids constructing a Gramian on the full Hilbert space and allows the analysis to focus on physically relevant terminal observables. 
For one-dimensional systems with local interactions, we presented an implementation based on matrix product states and time-evolving block decimation. 
The required sensitivities can be evaluated through forward propagation of the nominal state, backward propagation of observable states, local operator applications, and MPS inner-product contractions, without explicitly constructing the full state-transition matrix.

The numerical results illustrate three aspects of the proposed framework. 
First, the V-CAS can produce rankings that differ from static correlation- and entanglement-based centralities, reflecting its dependence on finite-time dynamical propagation and on the selected outputs. 
Second, high-scoring local Hamiltonian directions can provide useful candidates for the preselection of control channels in terminal-output control problems. 
Third, the tensor-network implementation reproduces the exact V-CAS for small systems and remains applicable to larger one-dimensional systems when the relevant states admit accurate MPS representations with moderate bond dimension.

Overall, the proposed framework provides an output- and trajectory-dependent allocation of virtual control intensities that can be interpreted as a measure of dynamical influence in quantum systems.Overall, the proposed framework provides an output- and trajectory-dependent allocation of virtual control intensities that can be interpreted as a measure of dynamical influence in quantum systems.
Its interpretation and effectiveness depend on the initial state, the terminal time, the virtual control Hamiltonians, and the selected observables. 
It should therefore be regarded as a model-based diagnostic for a prescribed control and measurement setting, rather than as a universal structural centrality of the drift Hamiltonian alone.
\subsection{Future Directions}

Several extensions of the present framework remain to be investigated.

First, a systematic analysis of the dependence of the V-CAS on the initial state, the terminal time, and the choice of output observables is needed. 
Because the allocation is defined along a nominal trajectory, changes in these quantities may alter both the reachable output directions and the relative allocation weights of the virtual control Hamiltonians.
Studying the resulting parameter dependence may help distinguish transient dynamical influence from behavior that persists over longer time horizons.

Second, the present formulation is based on the first-order response around a nominal trajectory. 
An important theoretical direction is to quantify the validity of this linearized description for finite perturbations. 
This includes deriving bounds on the difference between the nonlinear terminal-output response and its linear approximation, and clarifying when a large allocation weight assigned by the V-CAS to a control direction predicts effective finite-amplitude control.

Third, the continuous weight optimization considered in this paper should be connected more directly to sparse control-channel selection. 
In practical settings, only a fixed number of local controls may be available. 
Developing discrete or sparsity-promoting formulations, together with approximation guarantees for selecting a prescribed number of channels, would strengthen the use of the V-CAS as a control-placement criterion.

Fourth, extending the framework to open quantum systems is an important direction. 
For dynamics governed by Lindblad equations, both coherent controls and dissipative channels contribute to the terminal-output response. 
A controllability allocation score for such systems could provide a unified means of evaluating control resources, disturbance directions, and dissipation mechanisms in noisy quantum devices.

Finally, it would be useful to relate the proposed V-CAS to experimentally observable response quantities. 
Examples include the amplification of local perturbations, output fluctuations induced by control noise, and the dispersion of measured observables. 
Such connections may provide a more direct operational interpretation of controllability-based dynamical centrality and support its use in the design of quantum control and measurement architectures.
\section*{Acknowledgements}
This work was supported by JST PRESTO, Japan, Grant Number JPMJPR25K4.

\titleformat{\section}
  {\normalfont\Large\bfseries}
  {Appendix \thesection.}
  {0.5em}
  {}
\appendix
\section{Proofs for Section~II}\label{app:proofs-prelim}
\subsection{Well-posedness and Absolute Continuity}\label{proof:well-posed-AC}
\begin{proposition}[Well-posedness and absolute continuity under $L^2$ inputs]
\label{prop:wellposed-AC}
Suppose that Assumptions hold. For any initial condition $\psi_0\in\CC^d$, the Schr\"odinger equation
\begin{equation}
    \mathrm{i}\dot\psi(t)
    =
    \left(
        H_0+\sum_{j=1}^{m}u_j(t)H_j
    \right)\psi(t),
    \qquad
    \psi(0)=\psi_0,
\end{equation}
admits a unique absolutely continuous solution
\[
    \psi:[0,T]\to\mathbb C^d
\]
satisfying
\[
    \psi(t)
    =
    \psi_0+\int_0^t A(\tau)\psi(\tau)\,d\tau,
\]

where
\begin{equation}
    A(t)
    :=
    -\mathrm{i}
    \left(
        H_0+\sum_{j=1}^{m}u_j(t)H_j
    \right).
\end{equation}

Moreover, there exists a unique propagator $\Phi(t,s)$ such that, for every $s\in[0,T]$,
\begin{equation}
    \frac{\partial}{\partial t}\Phi(t,s)
    =
    A(t)\Phi(t,s)
    \quad\text{for a.e. }t,
    \qquad
    \Phi(s,s)=I,
\end{equation}
and
\begin{equation}
    \psi(t)=\Phi(t,s)\psi(s).
\end{equation}
The propagator is strongly continuous in $(t,s)$ and unitary. In
particular,
\begin{equation}
    \Phi(t,s)^*=\Phi(s,t),
    \qquad
    \|\psi(t)\|=\|\psi_0\|,
    \qquad
    t,s\in[0,T].
\end{equation}
Hence, if $\psi_0\in\mathcal S_{\CC^d}$, then
$\psi(t)\in\mathcal S_{\CC^d}$ for every $t\in[0,T]$.
\end{proposition}
\begin{proof}[Proof of Proposition~\ref{prop:wellposed-AC}]
Since $u_j\in L^2([0,T])\subset L^1([0,T])$ on the finite interval $[0,T]$, the matrix-valued function $A$ is measurable and satisfies
\begin{equation}
    \|A(t)\|
    \leq
    \|H_0\|
    +
    \sum_{j=1}^m |u_j(t)|\,\|H_j\|,
\end{equation}
with $\|A(\cdot)\|\in L^1([0,T])$. After extending $A$ outside $[0,T]$, if necessary, the standard existence and uniqueness theory for differential equations with measurable time dependence and integrable Lipschitz bounds yields a unique absolutely continuous solution; see Theorem~54 and Proposition~C.3.8 of Ref.~\cite{sontag1998mathematical}. 
The result applies to the present complex-valued system by identifying $\CC^d$ with $\RR^{2d}$. The solution satisfies the corresponding integral equation
\begin{equation}
    \psi(t)
    =
    \psi_0+\int_0^t A(\tau)\psi(\tau)\,\rd\tau.
\end{equation}

Applying the same result to
\begin{equation}
    \dot U(t)=A(t)U(t),
    \qquad
    U(0)=I,
\end{equation}
gives a unique absolutely continuous fundamental matrix $U(t)$.
Since the $H_j$ are Hermitian and the inputs $u_j$ are real-valued, $ A(t)^*=-A(t) $ for almost every $t$. Hence, using the product rule for absolutely
continuous matrix-valued functions,
\begin{align}
    \frac{\rd}{\rd t}\bigl(U(t)^*U(t)\bigr)
    &=
    U(t)^*
    \bigl(A(t)^*+A(t)\bigr)
    U(t)=0
\end{align}
for almost every $t$. Since $U(0)=I$, it follows that $ U(t)^*U(t)=I,$ so $U(t)$ is unitary.

Defining $    \Phi(t,s):=U(t)U(s)^*,$ then $\Phi(s,s)=I$ and
\begin{equation}
    \frac{\partial}{\partial t}\Phi(t,s)
    =
    A(t)\Phi(t,s)
\end{equation}
for almost every $t$. Uniqueness gives
$\psi(t)=\Phi(t,s)\psi(s)$. Moreover, since $U$ is continuous,
$\Phi$ is strongly continuous in $(t,s)$, and unitarity gives $ \Phi(t,s)^* = \Phi(s,t).$
Consequently, $ \|\psi(t)\| = \|\psi_0\|,$ which also proves the invariance of $\mathcal S_{\CC^d}$.
\end{proof}
\subsection{Proof of Proposition~\ref{prop:image-output-gramian}}\label{proof:image-output-gramian}
\begin{proof}
By Proposition~\ref{prop:wellposed-AC}, $\Phi_{\bar u}(T,t)$ is unitary and $\|\bar\psi(t)\|=1$. Since $C(T)$ is a linear map between finite-dimensional spaces, for every $j=1,\ldots,m$,
\begin{equation}
    \left\|
        C(T)\!\left[
            \Phi_{\bar u}(T,t)
            (-\mathrm{i}H_j)\bar\psi(t)
        \right]
    \right\|
    \leq
    \|C(T)\|\,\|H_j\|.
\end{equation}
Therefore, from the integral representation of $J_y$ and the
Cauchy--Schwarz inequality,
\begin{align}
    \|J_y[\delta u]\|
    &\leq
    \int_0^T
    \left\|
        \sum_{j=1}^m
        C(T)\!\left[
            \Phi_{\bar u}(T,t)
            (-\mathrm{i}H_j)\bar\psi(t)
        \right]
        \delta u_j(t)
    \right\|
    \,\rd t
    \notag\\
    &\leq
    \sqrt{T}\,
    \|C(T)\|
    \left(
        \sum_{j=1}^m \|H_j\|^2
    \right)^{1/2}
    \|\delta u\|_{L^2}.
\end{align}
Hence, $J_y: L^2([0,T];\mathbb R^m) \longrightarrow \mathbb R^r$ is a bounded linear operator.

Since $W(T)=J_yJ_y^*$, for every $z\in\mathbb R^r$,
\begin{equation}
    z^\top W(T)z
    =
    \langle J_y^*z,J_y^*z\rangle_{L^2}
    =
    \|J_y^*z\|_{L^2}^2.
\end{equation}
Thus,
\begin{equation}
    W(T)\succeq0,
    \qquad
    \ker W(T)=\ker J_y^*.
\end{equation}

By the standard range--kernel relation for adjoint operators \cite[Theorem~10.13(d)]{axler2020measure},
\begin{equation}
    \overline{\operatorname{Im}J_y}
    =
    (\ker J_y^*)^\perp.
\end{equation}
Since the codomain $\mathbb R^r$ is finite dimensional,
$\operatorname{Im}J_y$ is closed. Moreover, $W(T)$ is self-adjoint
on $\mathbb R^r$. Hence,
\begin{align}
    \operatorname{Im}W(T) =(\ker W(T))^\perp=(\ker J_y^*)^\perp=\operatorname{Im}J_y=\mathcal R_y(T).
\end{align}

Finally, since $W(T)\succeq0$,
\begin{align}
    W(T)\succ0
    \Longleftrightarrow
    \ker W(T)=\{0\}\Longleftrightarrow
    \operatorname{Im}W(T)=\mathbb R^r\Longleftrightarrow\operatorname{Im}J_y=\mathbb R^r.
\end{align}
The last condition is equivalent to the surjectivity of $J_y$ and,
by definition, to output controllability at time $T$ along the nominal
trajectory.
\end{proof}

\subsection{Proof of the Minimum Energy Property}\label{app:min-output-energy-proof}
\begin{proof}
By Proposition~\ref{prop:image-output-gramian}, $z\in\mathcal R_y(T)=\operatorname{Im}J_y$.
Since $J_y$ is bounded and its codomain $\mathbb R^r$ is finite dimensional, $\operatorname{Im}J_y$ is closed. Hence, the standard minimum-norm theorem for bounded operators between Hilbert spaces
\cite[Sec.~6.10, Theorem~1]{luenberger1997optimization} gives the unique
minimum-norm solution of $J_y[\delta u]=z$ as $ \delta u^\star=J_y^\dagger z.$  
Using $W(T)=J_yJ_y^*$ and the standard pseudoinverse identity \cite[Sec.~6.11]{luenberger1997optimization}, $ J_y^\dagger=J_y^*(J_yJ_y^*)^\dagger,$
we obtain $ \delta u^\star=J_y^*W(T)^+z.$

Finally, by Proposition~\ref{prop:image-output-gramian}, $z\in\operatorname{Im}W(T)$, so $    W(T)W(T)^+z=z.$
Therefore,
\begin{align}
    \|\delta u^\star\|_{L^2}^2
    &=
    \left\langle
        J_yJ_y^*W(T)^+z,\,
        W(T)^+z
    \right\rangle_{\mathbb R^r}
    =
    \left\langle
        W(T)W(T)^+z,\,
        W(T)^+z
    \right\rangle_{\mathbb R^r} =
    z^\top W(T)^+z.
\end{align}
Hence, \eqref{eq:minimum-energy-property} holds.

\end{proof}

\section{Proofs for Section~IV}
\paragraph{Storage cost of an MPS.}

We briefly explain why an MPS with bond dimension \(\chi\) has storage cost
linear in the system size and quadratic in \(\chi\).  The bond dimension
bounds the Schmidt rank across every bipartition of the one-dimensional
chain.  Consequently, at each site the state can be specified by data whose
size is proportional to the product of the left and right auxiliary
dimensions, both bounded by \(\chi\), together with the fixed local physical
dimension.  Since the local physical dimension is fixed for qubits, the
number of parameters per site is of order \(\chi^2\).  Summing over \(n\)
sites gives the storage scaling
\[
    \mathcal{O}(n\chi^2).
\]
In contrast, a generic \(n\)-qubit state vector requires storing \(2^n\)
complex amplitudes.  Therefore, when the relevant states admit accurate
approximations with moderate bond dimension, the MPS representation reduces
the storage cost compared with the full state-vector representation.

\subsection{Error Analysis for the MPS/TEBD Approximation}
\label{app:mps-error}

We prove the error bounds used in Section~\ref{subsec:mps-error}.
Throughout this subsection, we consider the uniform time grid
\[
    t_k=kh,
    \qquad
    k=0,\ldots,N,
    \qquad
    Nh=T.
\]
Since the nominal dynamics used to compute the V-CAS is generated by the time-independent drift Hamiltonian $H_0$, let
\[
    U:=e^{-\mathrm{i}hH_0}
\]
denote the exact one-step propagator.
Let $S$ denote the corresponding second-order symmetric
Trotter step.  We use the standard product-formula estimate
\begin{equation}
    \|U-S\|_2
    \leq
    C_{\mathrm{Tr}}h^3,
    \label{eq:app-local-trotter}
\end{equation}
where $C_{\mathrm{Tr}}>0$ is independent of $h$.
For example, if
\[
    H_0=H_{\mathrm A}+H_{\mathrm B},
    \qquad
    S=
    e^{-\mathrm{i}hH_{\mathrm A}/2}
    e^{-\mathrm{i}hH_{\mathrm B}}
    e^{-\mathrm{i}hH_{\mathrm A}/2},
\]
one may take
\begin{equation}
    C_{\mathrm{Tr}}
    =
    \frac{1}{12}
    \left\|
        [H_{\mathrm B},[H_{\mathrm B},H_{\mathrm A}]]
    \right\|_2
    +
    \frac{1}{24}
    \left\|
        [H_{\mathrm A},[H_{\mathrm A},H_{\mathrm B}]]
    \right\|_2 .
    \label{eq:app-trotter-constant}
\end{equation}

\subsubsection{Schmidt-truncation error}

We first recall the error associated with a single Schmidt
truncation.  Let
\[
    \phi
    =
    \sum_{\alpha=1}^{R}
    s_\alpha
    |L_\alpha\rangle |R_\alpha\rangle
\]
be a Schmidt decomposition, where the Schmidt coefficients are
ordered as $ s_1\geq s_2\geq\cdots\geq0.$

For a normalized forward state, define the discarded weight by $ \delta:=\sum_{\alpha>\chi}s_\alpha^2.$
The normalized rank-$\chi$ truncation is
\[
    \widetilde\phi_\chi
    :=
    \frac{
        \sum_{\alpha=1}^{\chi}
        s_\alpha
        |L_\alpha\rangle |R_\alpha\rangle
    }{\sqrt{1-\delta}}.
\]
Since $ \langle\phi,\widetilde\phi_\chi\rangle=\sqrt{1-\delta}, $ we obtain
\begin{align}
    \|\phi-\widetilde\phi_\chi\|_2^2
    &=
    2\left(1-\sqrt{1-\delta}\right)\leq
    2\delta.
    \label{eq:app-normalized-cut}
\end{align}

For the backward observable states, normalization is not
preserved and no renormalization is performed after truncation.
For a possibly unnormalized state $\phi$, define the absolute
discarded weight by
\begin{equation}
    \delta
    :=
    \sum_{\alpha>\chi}s_\alpha^2.
    \label{eq:app-unnormalized-discarded}
\end{equation}
If
\[
    P_\chi\phi
    :=
    \sum_{\alpha=1}^{\chi}
    s_\alpha
    |L_\alpha\rangle |R_\alpha\rangle,
\]
then orthogonality of the Schmidt vectors gives the exact identity
\begin{equation}
    \|\phi-P_\chi\phi\|_2
    =
    \sqrt{\delta}.
    \label{eq:app-unnormalized-cut}
\end{equation}
Moreover,
\begin{equation}
    \|P_\chi\phi\|_2
    \leq
    \|\phi\|_2.
    \label{eq:app-cut-contraction}
\end{equation}

For each forward TEBD step $q=1,\ldots,N$, let
$B_q^\psi$ denote the number of Schmidt truncations performed
during that step, and let
\[
    \delta_{q,b}^{\psi},
    \qquad
    b=1,\ldots,B_q^\psi,
\]
be their discarded weights.
Since the local gates between successive truncations are unitary,
repeated application of the triangle inequality and
\eqref{eq:app-normalized-cut} gives a total truncation error
during step $q$ bounded by
\begin{equation}
    \sum_{b=1}^{B_q^\psi}
    \sqrt{2\delta_{q,b}^{\psi}}.
    \label{eq:app-forward-one-step-cut}
\end{equation}

Similarly, for each backward step
$q=0,\ldots,N-1$ and each observable $l=1,\ldots,r$, let
$B_{q,l}^{z}$ denote the number of Schmidt truncations and let
\[
    \delta_{q,l,b}^{z},
    \qquad
    b=1,\ldots,B_{q,l}^{z},
\]
denote their absolute discarded weights.
The total truncation error during that backward step is bounded by
\begin{equation}
    \sum_{b=1}^{B_{q,l}^{z}}
    \sqrt{\delta_{q,l,b}^{z}}.
    \label{eq:app-backward-one-step-cut}
\end{equation}

Finally, let $B_{N,l}^{\mathrm{term}}$ denote the number of
truncations used when forming the terminal MPS
$O_l\widehat\psi_N$, and let
\[
    \delta_{N,l,b}^{\mathrm{term}},
    \qquad
    b=1,\ldots,B_{N,l}^{\mathrm{term}},
\]
denote their absolute discarded weights.

\subsubsection{Forward-state error}

The exact forward state satisfies
\begin{equation}
    \bar\psi_k
    =
    U\bar\psi_{k-1},
    \qquad
    k=1,\ldots,N.
    \label{eq:app-forward-exact}
\end{equation}
The approximate state $\widehat\psi_k$ is obtained by applying
the Trotter step $S$ to $\widehat\psi_{k-1}$ and performing the
Schmidt truncations described above.

Define
\[
    e_k^\psi
    :=
    \|\bar\psi_k-\widehat\psi_k\|_2,
    \qquad
    k=0,\ldots,N.
\]
Since both $\bar\psi_{k-1}$ and $\widehat\psi_{k-1}$ are
normalized, \eqref{eq:app-local-trotter} and
\eqref{eq:app-forward-one-step-cut} imply
\begin{align}
    e_k^\psi
    &\leq
    \|U\bar\psi_{k-1}-S\widehat\psi_{k-1}\|_2
    +
    \sum_{b=1}^{B_k^\psi}
    \sqrt{2\delta_{k,b}^{\psi}}
    \nonumber\\
    &\leq
    \|U(\bar\psi_{k-1}-\widehat\psi_{k-1})\|_2
    +
    \|(U-S)\widehat\psi_{k-1}\|_2
    +
    \sum_{b=1}^{B_k^\psi}
    \sqrt{2\delta_{k,b}^{\psi}}
    \nonumber\\
    &\leq
    e_{k-1}^\psi
    +
    C_{\mathrm{Tr}}h^3
    +
    \sum_{b=1}^{B_k^\psi}
    \sqrt{2\delta_{k,b}^{\psi}}.
    \label{eq:app-forward-recursion}
\end{align}
Since $e_0^\psi=0$, summing
\eqref{eq:app-forward-recursion} from $q=1$ to $q=k$ gives
\begin{align}
    e_k^\psi
    &\leq
    kC_{\mathrm{Tr}}h^3
    +
    \sum_{q=1}^{k}
    \sum_{b=1}^{B_q^\psi}
    \sqrt{2\delta_{q,b}^{\psi}}
    =
    t_kC_{\mathrm{Tr}}h^2
    +
    \sum_{q=1}^{k}
    \sum_{b=1}^{B_q^\psi}
    \sqrt{2\delta_{q,b}^{\psi}}.
    \label{eq:app-forward-total}
\end{align}
Therefore, with \eqref{eq:error}

we obtain the left part of \eqref{eq:backward-state-error-main}

\subsubsection{Backward-state error}

For each observable $O_l$, $l=1,\ldots,r$, let
\[
    L_l:=\|O_l\|_2.
\]
The exact terminal and backward states satisfy
\begin{align}
    z_{N,l}
    &=
    O_l\bar\psi_N,
    \label{eq:app-z-terminal-exact}\\
    z_{k,l}
    &=
    U^\dagger z_{k+1,l},
    \qquad
    k=0,\ldots,N-1.
    \label{eq:app-z-backward-exact}
\end{align}

The approximate terminal state $\widehat z_{N,l}$ is obtained
from $O_l\widehat\psi_N$ by MPS compression, without
renormalization.  Hence,
\begin{align}
    \|z_{N,l}-\widehat z_{N,l}\|_2
    &\leq
    \|O_l(\bar\psi_N-\widehat\psi_N)\|_2
    +
    \sum_{b=1}^{B_{N,l}^{\mathrm{term}}}
    \sqrt{\delta_{N,l,b}^{\mathrm{term}}}
    \leq
    L_l
    \left(
        E_N^{\psi,\mathrm{Tr}}
        +
        E_N^{\psi,\mathrm{cut}}
    \right)
    +
    \sum_{b=1}^{B_{N,l}^{\mathrm{term}}}
    \sqrt{\delta_{N,l,b}^{\mathrm{term}}}.
    \label{eq:app-terminal-z-error}
\end{align}

We next note that
\begin{equation}
    \|\widehat z_{k,l}\|_2
    \leq
    L_l,
    \qquad
    k=0,\ldots,N.
    \label{eq:app-z-norm}
\end{equation}
Indeed,
$\|\widehat\psi_N\|_2=1$ implies
\[
    \|O_l\widehat\psi_N\|_2
    \leq L_l.
\]
The terminal compression does not increase the norm by
\eqref{eq:app-cut-contraction}.  Each backward Trotter gate is
unitary, and each subsequent non-normalized Schmidt truncation is
also norm non-increasing.  Thus \eqref{eq:app-z-norm} follows
recursively.

Define
\[
    e_{k,l}^z
    :=
    \|z_{k,l}-\widehat z_{k,l}\|_2.
\]
For $k=0,\ldots,N-1$, we have
\begin{align}
    e_{k,l}^z
    &\leq
    \|U^\dagger z_{k+1,l}
       -S^\dagger\widehat z_{k+1,l}\|_2
    +
    \sum_{b=1}^{B_{k,l}^{z}}
    \sqrt{\delta_{k,l,b}^{z}}
    \nonumber\\
    &\leq
    \|U^\dagger
        (z_{k+1,l}-\widehat z_{k+1,l})\|_2
    +
    \|(U^\dagger-S^\dagger)
        \widehat z_{k+1,l}\|_2
    +
    \sum_{b=1}^{B_{k,l}^{z}}
    \sqrt{\delta_{k,l,b}^{z}}
    \nonumber\\
    &\leq
    e_{k+1,l}^z
    +
    L_lC_{\mathrm{Tr}}h^3
    +
    \sum_{b=1}^{B_{k,l}^{z}}
    \sqrt{\delta_{k,l,b}^{z}},
    \label{eq:app-backward-recursion}
\end{align}
where we used $ \|U^\dagger-S^\dagger\|_2=\|U-S\|_2$ together with \eqref{eq:app-z-norm}.

Summing \eqref{eq:app-backward-recursion} from
$q=k$ to $q=N-1$ and using
\eqref{eq:app-terminal-z-error} yields
\begin{align}
    e_{k,l}^z
    &\leq
    L_lE_N^{\psi,\mathrm{Tr}}
    +
    (N-k)L_lC_{\mathrm{Tr}}h^3
    +
    L_lE_N^{\psi,\mathrm{cut}}
    +
    \sum_{b=1}^{B_{N,l}^{\mathrm{term}}}
    \sqrt{\delta_{N,l,b}^{\mathrm{term}}}
    +
    \sum_{q=k}^{N-1}
    \sum_{b=1}^{B_{q,l}^{z}}
    \sqrt{\delta_{q,l,b}^{z}}.
    \label{eq:app-backward-total-1}
\end{align}
Since
\[
    E_N^{\psi,\mathrm{Tr}}
    =
    TC_{\mathrm{Tr}}h^2,
    \qquad
    (N-k)h=T-t_k,
\]
the Trotter part becomes
\begin{align}
    L_lE_N^{\psi,\mathrm{Tr}}
    +(N-k)L_lC_{\mathrm{Tr}}h^3
    &=
    L_l(2T-t_k)C_{\mathrm{Tr}}h^2.
\end{align}
Thus, defining \eqref{eq:backward-trotter-budget-main}, \eqref{eq:backward-cut-budget-main}, we obtain the right part of \eqref{eq:backward-state-error-main}.

\subsubsection{Sensitivity error}

Let
\[
    a_j:=\|H_j\|_2,
    \qquad
    j=1,\ldots,m.
\]
Using the exact and approximate sensitivities,
\begin{align*}
    (y_{j,k})_l
    &=
    2\operatorname{Re}
    \langle
        z_{k,l},
        -\mathrm{i}H_j\bar\psi_k
    \rangle,\quad
    (\widehat y_{j,k})_l
    =
    2\operatorname{Re}
    \langle
        \widehat z_{k,l},
        -\mathrm{i}H_j\widehat\psi_k
    \rangle,
\end{align*}
we obtain
\begin{align}
    |(y_{j,k})_l-(\widehat y_{j,k})_l|
    \leq
    2
    \left|
        \left\langle
            z_{k,l}-\widehat z_{k,l},
            -\mathrm{i}H_j\bar\psi_k
        \right\rangle
    \right|+2
    \left|
        \left\langle
            \widehat z_{k,l},
            -\mathrm{i}H_j
            (\bar\psi_k-\widehat\psi_k)
        \right\rangle
    \right|
    \leq
    2a_j
    \left(
        e_{k,l}^{z}
        +
        L_le_k^\psi
    \right).
    \label{eq:app-component-sensitivity}
\end{align}
Therefore,
\begin{align}
    |(y_{j,k})_l-(\widehat y_{j,k})_l|
    &\leq
    \eta_{j,k,l}^{\mathrm{Tr}}
    +
    \eta_{j,k,l}^{\mathrm{cut}},
    \label{eq:app-component-sensitivity-budget}
\end{align}
where $\eta_{j,k,l}^{\mathrm{Tr}}$ and $\eta_{j,k,l}^{\mathrm{cut}}$ are defined by \eqref{eq:error-eta}.
By triangle inequality in $\RR^r$, we have \eqref{eq:sensitivity-error-main}.

\subsubsection{Proof of Theorem~\ref{thm:mps-gramian-error}}

For any $x,y\in\RR^r$,
\begin{align}
    xx^\top-yy^\top
    &=
    (x-y)y^\top
    +
    x(x-y)^\top.
\end{align}
Hence,
\begin{align}
    \|xx^\top-yy^\top\|_2
    \leq
    \|x-y\|_2\|y\|_2
    +
    \|x\|_2\|x-y\|_2
    =
    (\|x\|_2+\|y\|_2)
    \|x-y\|_2.
    \label{eq:app-outer-product}
\end{align}

Using the definitions of the discrete Gramians in
Section~\ref{practical-construction},
\begin{equation}
    \begin{split}
        \|W_j^h(T)-\widehat W_j^h(T)\|_2
&\le
h\sum_{k=0}^{N}\omega_k
\left\|
y_{j,k}y_{j,k}^{\top}
-
\widehat y_{j,k}\widehat y_{j,k}^{\top}
\right\|_2
\\
&\le
h\sum_{k=0}^{N}\omega_k
\bigl(
\|y_{j,k}\|_2+\|\widehat y_{j,k}\|_2
\bigr)
\|y_{j,k}-\widehat y_{j,k}\|_2.
    \end{split}\label{proof:differential-of-W}
\end{equation}

Equation~\eqref{eq:sensitivity-error-main} implies
\[
    \|y_{j,k}\|_2
    \leq
    \|\widehat y_{j,k}\|_2
    +
    \eta_{j,k}^{\mathrm{Tr}}
    +
    \eta_{j,k}^{\mathrm{cut}}.
\]
Substituting it into \eqref{proof:differential-of-W} gives \eqref{eq:gramian-error-decomposition-main}.

\section{Numerical Experiment Settings}
\label{sec:numerical_experiment_settings}

This appendix provides the definitions of the static centrality measures and the detailed model parameters used in the comparison in Sec.~\ref{subsec:compare}.
All models consist of \(n=8\) qubits, with Pauli operators \(\sigma_i^x\), \(\sigma_i^y\), and \(\sigma_i^z\) acting on site \(i\).
The initial state is
\begin{equation}
    \bar\psi(0)=|0\rangle^{\otimes 8},
\end{equation}
and the uncontrolled dynamics are evolved up to \(T=0.6\) with time step \(\Delta t=0.02\).
For the V-CAS comparison, the local virtual control Hamiltonians and observables are chosen as
\begin{equation}
    H_j=O_j=\sigma_j^z,
    \qquad j=1,\ldots,8.
\end{equation}

\paragraph{Static centrality measures.}
The static centralities are evaluated from the uncontrolled terminal state
\(\bar\psi(T)\). Define
\[
    \rho(T)=|\bar\psi(T)\rangle\langle\bar\psi(T)|,
    \qquad
    \rho_i=\operatorname{Tr}_{\{i\}^c}\rho(T),
    \qquad
    \rho_{ij}=\operatorname{Tr}_{\{i,j\}^c}\rho(T).
\]
With the von Neumann entropy
\(S(\rho)=-\operatorname{Tr}(\rho\log\rho)\), the mutual information is
\[
    I_{ij}=S(\rho_i)+S(\rho_j)-S(\rho_{ij}),
\]
and the mutual-information strength centrality
\cite{sundar2018complex} is
\[
    c_i^{\mathrm{MI}}=\sum_{j\neq i} I_{ij}.
\]

For concurrence strength \cite{nokkala2024complex}, let
\[
    \widetilde{\rho}_{ij}
    =
    (\sigma^y\otimes\sigma^y)\rho_{ij}^{*}
    (\sigma^y\otimes\sigma^y).
\]
If
\(\lambda_1\geq\lambda_2\geq\lambda_3\geq\lambda_4\)
are the square roots of the eigenvalues of
\(\rho_{ij}\widetilde{\rho}_{ij}\), define
\[
    C_{ij}
    =
    \max\{0,\lambda_1-\lambda_2-\lambda_3-\lambda_4\},
    \qquad
    c_i^{\mathrm{C}}
    =
    \sum_{j\neq i}C_{ij}.
\]

For mutual-information betweenness, the mutual-information network is assigned edge lengths
\[
    \ell_{ij}=\frac{1}{I_{ij}+\varepsilon},
    \qquad
    \varepsilon=10^{-15},
\]
for \(I_{ij}>\varepsilon\), with the remaining edges omitted.
The standard weighted betweenness centrality of this graph is denoted by \(c_i^{\mathrm{B}}\).
All three centrality vectors are normalized to unit \(\ell_1\) norm before computing the correlations in Table~\ref{tab:static-centrality-correlation}.

We consider four representative systems: a central-spin model, a small-world spin network, a modular two-cluster model, and a disordered XXZ transport chain.
The first three models use random seed \(1\).
For the disordered XXZ model, the realization is generated with random seed \(3\), as specified below.
The model parameters are fixed throughout the centrality comparison.

We also verified finite-time feasibility and uniqueness of the V-CAS for all four systems used in Table~\ref{tab:static-centrality-correlation}.
For each system, let
\[
D(T)
=
\begin{bmatrix}
\operatorname{vech}(W_1(T)-W_8(T)) &
\cdots &
\operatorname{vech}(W_7(T)-W_8(T))
\end{bmatrix}.
\]
At \(T=0.6\), the total Gramian is positive definite and
\[
    \operatorname{rank}D(T)=7
\]
for all four systems.
The smallest eigenvalues of the total Gramians are \(7.42\times10^{-9}\), \(6.50\times10^{-4}\), \(1.30\times10^{-3}\), and \(2.05\times10^{-5}\) for the central-spin, small-world, modular two-cluster, and
disordered XXZ systems, respectively.
The corresponding smallest singular values of \(D(T)\) are \(1.61\times10^{-8}\), \(7.67\times10^{-4}\), \(1.70\times10^{-3}\), and \(2.50\times10^{-4}\).
Thus, the allocation problem is feasible and the optimal allocation is unique for every numerical instance used in the static-centrality comparison.

\subsection{Models for Comparison With Existing Centrality Metrics}

\subsubsection{Central-Spin Model (\texttt{central\_spin})}

The central-spin model consists of a central spin at site \(1\) strongly coupled to the remaining seven bath spins.
The Hamiltonian is
\begin{equation}
\begin{aligned}
H_0
={}&
\sum_{i=1}^{8} h_i \sigma_i^z
+
\sum_{i=2}^{8}
J_{1i}
\left(
\sigma_1^x\sigma_i^x
+
\sigma_1^y\sigma_i^y
+
\sigma_1^z\sigma_i^z
\right)
+
\sum_{\substack{2\le i<j\le8\\(i,j)\in\mathcal E_{\rm bath}}}
J_{ij}^{\rm bath}\,
\sigma_i^x\sigma_j^x .
\end{aligned}
\end{equation}
The local fields are generated as
\begin{equation}
    h_i
    =
    0.5
    \left(
        1+0.35\,\xi_i^{(h)}
    \right),
\end{equation}
where
\(\xi_i^{(h)}\sim\mathcal N(0,1)\).
The central-spin couplings are
\begin{equation}
    J_{1i}
    =
    4.0
    \left(
        1+0.20\,\xi_i^{(c)}
    \right),
    \qquad
    i=2,\ldots,8,
\end{equation}
with
\(\xi_i^{(c)}\sim\mathcal N(0,1)\).

For every pair of bath spins, a bath--bath edge is independently included with probability
\[
    p_{\rm bath}=0.25.
\]
For each included edge,
\begin{equation}
    J_{ij}^{\rm bath}
    =
    0.08\,\xi_{ij}^{(b)},
    \qquad
    \xi_{ij}^{(b)}\sim\mathcal N(0,1).
\end{equation}
Thus, the system contains a dominant central-spin interaction together with weak random couplings within the bath.

\subsubsection{Small-World Model (\texttt{small\_world})}

The small-world model is constructed from a periodic ring in which each site is initially connected to its first- and second-nearest neighbors. 
Each edge is independently rewired with probability
\begin{equation}
    p_{\rm rewire}=0.25.
\end{equation}

The Hamiltonian is
\begin{equation}
H_0
=
\sum_{i=1}^{8}
\left(
h_i^{(z)}\sigma_i^z
+
0.4\,\sigma_i^x
\right)
+
\sum_{(i,j)\in\mathcal E_{\rm SW}}
J_{ij}
\left(
\sigma_i^z\sigma_j^z
+
0.15\,\sigma_i^x\sigma_j^x
\right),
\end{equation}
where \(\mathcal E_{\rm SW}\) denotes the edge set after rewiring.

The longitudinal fields are generated as
\begin{equation}
    h_i^{(z)}
    =
    0.25
    \left(
        1+0.25\,\xi_i^{(h)}
    \right),
    \qquad
    \xi_i^{(h)}\sim\mathcal N(0,1),
\end{equation}
and the interaction strengths are
\begin{equation}
    J_{ij}
    =
    1.0
    \left(
        1+0.15\,\xi_{ij}^{(J)}
    \right),
    \qquad
    \xi_{ij}^{(J)}\sim\mathcal N(0,1).
\end{equation}
This construction combines predominantly Ising-type interactions with weak transverse \(XX\) couplings and nonlocal shortcuts generated by the small-world rewiring.

\subsubsection{Modular Two-Cluster Model (\texttt{modular\_2c})}

The eight spins are divided into two communities,
\begin{equation}
    \mathcal A=\{1,2,3,4\},
    \qquad
    \mathcal B=\{5,6,7,8\}.
\end{equation}
For a pair of sites belonging to the same community, an edge is included independently with probability $ p_{\rm in}=0.7,$ whereas an edge connecting the two communities is included with probability $   p_{\rm out}=0.12.$

The Hamiltonian is
\begin{equation}
H_0
=
\sum_{i=1}^{8}
\left(
h_i^{(z)}\sigma_i^z
+
0.35\,\sigma_i^x
\right)
+
\sum_{(i,j)\in\mathcal E_{\rm mod}}
J_{ij}
\left(
\sigma_i^z\sigma_j^z
+
0.12\,\sigma_i^x\sigma_j^x
\right),
\end{equation}
where \(\mathcal E_{\rm mod}\) is the randomly generated modular edge set.

The longitudinal fields are
\begin{equation}
    h_i^{(z)}
    =
    0.25
    \left(
        1+0.4\,\xi_i^{(h)}
    \right),
    \qquad
    \xi_i^{(h)}\sim\mathcal N(0,1).
\end{equation}
For an included edge, the coupling is generated as
\begin{equation}
J_{ij}
=
\begin{cases}
1.2
\left(
1+0.2\,\xi_{ij}^{(J)}
\right),
&
i,j\in\mathcal A
\ \text{or}\
i,j\in\mathcal B,
\\[2mm]
0.25
\left(
1+0.2\,\xi_{ij}^{(J)}
\right),
&
i\in\mathcal A,\ j\in\mathcal B,
\end{cases}
\end{equation}
where
\(\xi_{ij}^{(J)}\sim\mathcal N(0,1)\).
Thus, both the connection probability and the characteristic interaction strength are larger within each community than between the two communities.

\subsubsection{Disordered XXZ Transport Model
(\texttt{disordered\_xxz\_transport})}

As a fourth example, we consider a one-dimensional disordered XXZ spin chain with nearest-neighbor interactions. The Hamiltonian is
\begin{equation}
\begin{aligned}
H_0
={}&
\sum_{i=1}^{8}
h_i\sigma_i^z
+
h_x
\sum_{i=1}^{8}
\sigma_i^x
+
\sum_{i=1}^{7}
J_i
\left(
\sigma_i^x\sigma_{i+1}^x
+
\sigma_i^y\sigma_{i+1}^y
+
\Delta\,
\sigma_i^z\sigma_{i+1}^z
\right).
\end{aligned}
\end{equation}
We set
\begin{equation}
    J=1,\qquad
    \Delta=1,\qquad
    h_x=0.25.
\end{equation}
The longitudinal fields are independently sampled from
\begin{equation}
    h_i\sim\operatorname{Unif}[-1.5,1.5],
    \qquad i=1,\ldots,8.
\end{equation}
The nearest-neighbor coupling strengths contain independent Gaussian disorder,
\begin{equation}
    J_i
    =
    1.0
    \left(
        1+0.25\,\xi_i^{(J)}
    \right),
    \qquad
    \xi_i^{(J)}\sim\mathcal N(0,1),
    \qquad i=1,\ldots,7.
\end{equation}
The realization used in the numerical comparison is generated with random seed \(3\).

This model provides a spatially local but strongly inhomogeneous example in which finite-time perturbation propagation is affected by both random on-site detunings and disordered exchange couplings.
It is therefore used to examine whether the proposed V-CAS identifies influence patterns that differ from those inferred from static correlation- and entanglement-based measures.

\subsection{Hard-Core-Boson Benchmark Details}
\label{app:hcb_benchmark}

This appendix provides the numerical details of the hard-core-boson benchmark in Sec.~\ref{subsec:control_experiment}.

\subsubsection{Physical models and output coordinates}

We consider a half-filled hard-core-boson chain of length
\[
    L=10,
    \qquad
    N_{\mathrm p}=5.
\]
Since the total particle number is conserved, the dynamics remain in the fixed-particle-number sector, whose Hilbert-space dimension is
\[
    \dim\mathcal H
    =
    \binom{10}{5}
    =
    252.
\]

The drift Hamiltonian is
\begin{equation}
    H_0
    =
    -\sum_{j=1}^{L-1}
    t_j
    \left(
        b_j^\dagger b_{j+1}
        +
        b_{j+1}^\dagger b_j
    \right)
    +
    V
    \sum_{j=1}^{L-1}
    n_jn_{j+1}
    +
    \sum_{j=1}^{L}
    h_j n_j,
    \label{eq:app_hcb_hamiltonian}
\end{equation}
where
\[
    n_j=b_j^\dagger b_j,
    \qquad
    V=1.
\]
For both geometries, the hopping amplitude is $t_j=1$.

A weak on-site random potential is added to each geometry:
\begin{equation}
    h_j
    =
    h_j^{\mathrm{geom}}
    +
    \delta_j,
    \qquad
    \delta_j
    =
    0.03
    \left(
        \xi_j-\frac12
    \right),
    \label{eq:app_hcb_disorder}
\end{equation}
where the \(\xi_j\) are generated independently from
\(\operatorname{Unif}[0,1)\).
The same disorder realization, generated using system seed $1$, is used for both geometries.
We consider two Hamiltonian configurations.
For the barrier geometry,
\[
    h_5^{\mathrm{geom}}
    =
    h_6^{\mathrm{geom}}
    =
    0.90,
\]
and all other deterministic on-site potentials vanish.

For the harmonically confined geometry,
\[
    h_j^{\mathrm{geom}}
    =
    0.10
    \left(
        j-\frac{L+1}{2}
    \right)^2.
\]

The virtual and physical control Hamiltonians are the local density potentials
\[
    H_j=n_j,
    \qquad
    j=1,\ldots,L.
\]
The nominal initial state is the charge-density-wave state
\[
    \bar\psi(0)
    =
    \ket{1010101010}.
\]

The physical density profile is
\[
    \rho(\psi)
    =
    \begin{bmatrix}
        \langle n_1\rangle_\psi&
        \cdots&
        \langle n_L\rangle_\psi
    \end{bmatrix}^{\top}.
\]
Since the total particle number is conserved, its uniform component contains no independent output information.
We therefore use
\begin{equation}
    y(\psi)=B\rho(\psi),
    \label{eq:app_hcb_output}
\end{equation}
where
\[
    B\boldsymbol{1}=0,
    \qquad
    BB^{\top}=I_{L-1}.
\]
Numerically, \(B\) is chosen as the Helmert zero-sum basis.
More explicitly, for \(q=1,\ldots,L-1\),
\[
    B_{qj}
    =
    \begin{cases}
        \displaystyle
        \frac{1}{\sqrt{q(q+1)}},
        & 1\le j\le q,
        \\[2mm]
        \displaystyle
        -\frac{q}{\sqrt{q(q+1)}},
        & j=q+1,
        \\[2mm]
        0,
        & j>q+1.
    \end{cases}
\]
Thus, the output dimension is
\[
    r=L-1=9.
\]

\subsubsection{Gramian and placement computation}

For each geometry, the individual Gramian contributions \(W_j(T)\), \(j=1,\ldots,L\), are computed along the uncontrolled trajectory with \(T=2\).
The Gramian integral is approximated by the trapezoidal rule using
\[
    M_{\mathrm{score}}=40
\]
uniform intervals, corresponding to
\[
    \Delta t_{\mathrm{score}}=0.05.
\]
Each quadrature interval is resolved by two second-order TEBD substeps.
The MPS calculations use maximum bond dimension
\(\chi_{\max}=24\) and singular-value cutoff \(10^{-10}\).

The weighted Gramian is
\begin{equation}
    W(p,T)
    =
    \sum_{j=1}^{L}
    p_jW_j(T).
    \label{eq:app_hcb_weighted_gramian}
\end{equation}
No additive regularization is applied to the individual Gramians, the weighted Gramian, or the subset Gramians.

The V-CAS is obtained by solving \eqref{opt:controllability} using an Armijo projected-gradient method based on \cite[Algorithm~1]{sato2026projected}.
The iteration starts from the uniform allocation, and positive definiteness is checked directly during the Armijo search.

The score-based two-site placement is
\[
    S_{\mathrm{score}}
    =
    \operatorname{top}_2(p^\star).
\]

For comparison, for each two-site subset
\(S\subset\{1,\ldots,L\}\), we define the equal-allocation Gramian
\begin{equation}
    W_S(T)
    =
    \frac{1}{2}
    \sum_{j\in S}W_j(T).
    \label{eq:app_hcb_subset_gramian}
\end{equation}
The discrete log-determinant placement is
\[
    S_{\mathrm D}
    \in
    \operatorname*{arg\,max}_{|S|=2}
    \log\det W_S(T),
\]
where a subset with a non-positive-definite Gramian lies outside the log-determinant domain.
The trace placement maximizes
\(\operatorname{Tr}W_S(T)\) over the same
\(\binom{10}{2}=45\) candidate pairs.

The physical heuristic is also target independent.
For both geometries, it selects the two central sites, $S_{physical} = \{5,6\}$.

The random pairs are generated with random seed \(300\).
All placement rules are determined once for each geometry and are independent of the terminal targets.

\subsubsection{Target generation}
\label{app:hcb_targets}

Terminal targets are generated using all \(L=10\) local density controls.
The teacher controls use
\[
    M_{\mathrm{control}}=12
\]
piecewise-constant time intervals.

We use five target families,
\[
    \mathcal F
    =
    \{
        \mathrm{global},
        \mathrm{center},
        \mathrm{edges},
        \mathrm{staggered},
        \mathrm{tilt}
    \}.
\]
For each family, one teacher-control direction is generated.

For the global family, independent standard-normal random variables are generated for every site and every control interval.
Each site's temporal sequence is subsequently smoothed by a moving-average window of length three.

For the other four families, a single random temporal profile \(a_q\) is generated and smoothed in the same manner, and the site-dependent teacher direction is constructed as
\[
    d_{qj}=a_q g_j.
\]

For the center family,
\[
    g_j^{\mathrm{center}}
    =
    \exp
    \left[
        -\frac12
        \left(
            \frac{j-(L+1)/2}{w_{\mathrm c}}
        \right)^2
    \right],
    \qquad
    w_{\mathrm c}
    =
    \max\left\{0.75,\frac{L}{10}\right\}.
\]

For the edge family,
\[
    g_j^{\mathrm{edge}}
    =
    \exp
    \left[
        -\frac12
        \left(
            \frac{j-1}{w_{\mathrm e}}
        \right)^2
    \right]
    +
    \exp
    \left[
        -\frac12
        \left(
            \frac{j-L}{w_{\mathrm e}}
        \right)^2
    \right],
\]
where
\[
    w_{\mathrm e}
    =
    \max\left\{0.55,\frac{L}{14}\right\}.
\]

For the staggered family,
\[
    g_j^{\mathrm{stag}}
    =
    (-1)^{j-1},
\]
and for the tilted family,
\[
    g_j^{\mathrm{tilt}}
    =
    -1
    +
    \frac{2(j-1)}{L-1},
    \qquad
    j=1,\ldots,L.
\]

Each generated teacher direction is normalized such that
\[
    \max_{q,j}|d_{qj}|=1.
\]

For system index \(s=1,2\) and family index \(f=1,\ldots,5\), the target-generation random seed is
\begin{equation}
    s_{s,f}
    =
    101+1000s+100f,
    \label{eq:app_hcb_target_seed}
\end{equation}
where the family order is
\[
    \mathrm{global},
    \mathrm{center},
    \mathrm{edges},
    \mathrm{staggered},
    \mathrm{tilt}.
\]

For each family, five candidate teacher directions are generated and first evaluated at the reference amplitude
\[
    a_{\mathrm{ref}}=0.30.
\]
Let \(y_0(T)\) denote the uncontrolled terminal output.
A candidate is preferred when
\[
    0.02
    \le
    \left\|
        y_{\mathrm{teacher}}(T)-y_0(T)
    \right\|_2
    \le
    0.14.
\]
Among candidates with the same feasibility status, the direction whose terminal-output displacement is closest to the midpoint \(0.08\) is selected.

After the teacher direction has been selected, exactly the same direction is scaled to the three amplitudes
\[
    0.15,\qquad0.30,\qquad0.50.
\]
This construction avoids confounding the target-amplitude comparison with a change of target direction.

Each geometry therefore contains
\[
    5\times1\times3=15
\]
target--amplitude cases, giving \(30\) cases over the two geometries.

\subsubsection{Nonlinear control optimization}
\label{app:hcb_grape}

For each target, the nonlinear terminal-output control problem is solved for the two-site placements returned by the target-independent selection rules described above.
If multiple selection rules return the same site pair, that pair is optimized only once.
Thus, the nonlinear optimization is not performed exhaustively over all \(\binom{10}{2}=45\) possible two-site placements.

The control horizon is divided into
\[
    M_{\mathrm{control}}=12
\]
piecewise-constant intervals, so that
\[
    h
    =
    \frac{T}{M_{\mathrm{control}}}
    =
    \frac{1}{6}.
\]
Each control interval is resolved by four second-order TEBD substeps.
The MPS calculations use
\(\chi_{\max}=24\) and singular-value cutoff \(10^{-10}\).
The controls satisfy
\[
    |u_{q,j}|\le0.5.
\]

For a selected pair \(S\), the nonlinear optimization objective is
\begin{equation}
    J(u)
    =
    \frac12
    \left\|
        y(T;u)-y_{\mathrm{tar}}
    \right\|_2^2
    +
    \frac{\alpha h}{2}
    \sum_{q=0}^{M_{\mathrm{control}}-1}
    \sum_{j\in S}
    u_{q,j}^2,
    \label{eq:app_hcb_grape_objective}
\end{equation}
with
\[
    \alpha=10^{-5}.
\]
The parameter \(\alpha\) is only a quadratic control-energy penalty in the nonlinear control problem and is unrelated to the controllability-Gramian computation.

The objective is minimized with L-BFGS-B using a forward MPS/TEBD trajectory and an adjoint MPS/TEBD sweep.
The gradient differentiates the control-dependent on-site gates in the chosen TEBD factorization; derivatives of the SVD-compression map are not included in the adjoint calculation.
The numerical optimization settings are
\[
    \mathrm{maxiter}=60,
    \qquad
    \mathrm{ftol}=10^{-12},
    \qquad
    \mathrm{gtol}=10^{-7},
    \qquad
    \mathrm{maxls}=30.
\]
Each optimization is initialized from the zero-control input, with no random restarts.

Performance is reported using the normalized terminal error
\[
    e_{\mathrm{norm}}
    =
    \frac{
        \|y(T;u^\star)-y_{\mathrm{tar}}\|_2
    }{
        \|y_0(T)-y_{\mathrm{tar}}\|_2
    }.
\]

\subsection{Finite-Time Uniqueness Diagnostics for the HCB Benchmark}
\label{app:hcb_uniqueness}

For the hard-core-boson experiment, the initial-response sufficient condition of Proposition~\ref{prop-unique} is not applicable.
The CDW initial state is an eigenstate of every local number operator
\(n_j\), and hence
\[
    C(0)
    \left[
        -\mathrm{i}n_j\bar\psi(0)
    \right]
    =
    0,
    \qquad
    j=1,\ldots,L.
\]

We therefore test affine independence of the finite-time Gramian contributions directly at \(T=2\).
Since the output dimension is \(r=9\), \(\operatorname{Sym}_9\) has dimension \(45\).
For each geometry, define
\begin{equation}
    D(T)
    =
    \begin{bmatrix}
        \operatorname{vech}
        \left(
            W_1(T)-W_{10}(T)
        \right)
        &
        \cdots
        &
        \operatorname{vech}
        \left(
            W_9(T)-W_{10}(T)
        \right)
    \end{bmatrix}
    \in
    \mathbb{R}^{45\times9}.
    \label{eq:app_hcb_affine_matrix}
\end{equation}
The ten matrices \(W_1(T),\ldots,W_{10}(T)\) are affinely independent if and only if
\[
    \operatorname{rank}D(T)=9.
\]

Numerical rank is evaluated using the threshold
\begin{equation}
    \tau_{\mathrm{rank}}
    =
    10^{-10}
    \sigma_{\max}(D(T)).
    \label{eq:app_hcb_rank_threshold}
\end{equation}
For both HCB geometries,
\[
    \operatorname{rank}_{\tau_{\mathrm{rank}}}
    D(T)
    =
    9.
\]
The smallest singular value of \(D(T)\) over the two geometries lies in the interval
\[
    7.49\times10^{-3}
    \le
    \sigma_{\min}(D(T))
    \le
    1.15\times10^{-2},
\]
which is well above the numerical rank threshold.

We also verify feasibility of the log-determinant problem by evaluating the total Gramian
\[
    W_{\mathrm{tot}}(T)
    =
    \sum_{j=1}^{L}W_j(T).
\]
Positive definiteness is checked directly from its eigenvalues using the numerical threshold
\[
    \tau_{\mathrm{PD}}
    =
    10^{-10}
    \max
    \left\{
        1,\,
        \max_\ell
        |\lambda_\ell(W_{\mathrm{tot}}(T))|
    \right\}.
\]
For all two geometries,
\[
    \lambda_{\min}
    \left(
        W_{\mathrm{tot}}(T)
    \right)
    >
    \tau_{\mathrm{PD}}.
\]

Hence, although the initial-response sufficient condition vanishes for the CDW state, both HCB geometries satisfy both finite-time feasibility and affine independence of the Gramian contributions at \(T=2\).
Accordingly, the numerical instances lie in the regime where the finite-time V-CAS is well defined and unique.

\subsection{MPS Accuracy and Runtime Benchmark}
\label{app:mps_benchmark}

This appendix provides the detailed settings for the MPS accuracy and runtime benchmark in Sec.~\ref{sec:mps_benchmark}.
The drift Hamiltonian is
\begin{equation}
    H_0
    =
    \sum_{j=1}^{n}
    \left(
        h_x X_j+h_j^z Z_j
    \right)
    +
    \sum_{j=1}^{n-1}
    J_j Z_j Z_{j+1},
\end{equation}
where
\[
    h_x=0.8,
\]
and the site-dependent longitudinal fields and nearest-neighbor couplings are generated as
\[
    h_j^z
    =
    0.2\left(1+0.5\xi_j\right),
    \qquad
    J_j
    =
    1.0\left(1+0.1\eta_j\right),
\]
with independent standard-normal random variables \(\xi_j,\eta_j\sim\mathcal{N}(0,1)\).
The realization is generated with random seed \(1\).

The initial state is
\[
    \bar{\psi}(0)=\ket{0}^{\otimes n},
\]
and the local virtual Hamiltonians and observables are
\[
    H_j=O_j=Z_j,
    \qquad j=1,\ldots,n.
\]
The V-CAS is evaluated over the horizon
\[
    T=0.6
\]
using the uniform time step
\[
    \Delta t=0.02.
\]
No additive regularization is applied in the V-CAS optimization.

For the MPS implementation, we use
\[
    \chi_{\max}=10,
    \qquad
    \varepsilon_{\mathrm{SVD}}=10^{-10}.
\]
Exact full-Hilbert-space calculations are performed for \(n=2,\ldots,10\), whereas the MPS calculations are performed for \(n=6,8,\ldots,30\) in increments of two.
The accuracy comparison is therefore performed on the common system sizes \(n=6,8,10\).

For each system size, the runtime is measured three times and the mean and standard deviation are reported.
One MPS warm-up calculation is performed before the timing measurements.
All calculations are performed within the same Windows 11 environment using Python 3.14.2 and NumPy 2.4.2 on an Intel x86-64 processor.




\bibliographystyle{iopart-num}
\bibliography{reference}



%




\end{document}